\documentclass[aos,preprint]{imsart}

\RequirePackage[colorlinks,citecolor=blue,urlcolor=blue]{hyperref}

\usepackage[english]{babel}
\usepackage{stmaryrd}
\usepackage{cancel}
\usepackage{Clustered_Joint_Processing}
\usepackage{mathtools}
\usepackage{amsmath,amssymb,amsthm,mathrsfs,amsfonts,dsfont} 
\usepackage{abstract}
\usepackage{bbold}

\newtheorem{remark}{Remark}
\newtheorem{definition}{Definition}

\newcommand{\cX}{{\mathcal X}}
\newcommand{\cY}{{\mathcal Y}}
\newcommand{\cU}{{\mathcal U}}
\newcommand{\cV}{{\mathcal V}}
\newcommand{\cT}{{\mathcal T}}
\newcommand{\cP}{{\mathcal P}}
\newcommand{\cE}{{\mathcal E}}
\newcommand{\cC}{{\mathcal C}}
\newcommand{\cF}{{\mathcal F}}
\newcommand{\cA}{{\mathcal A}}
\newcommand{\cD}{{\mathcal D}}

\newcommand{\bX}{{\bar X}}
\newcommand{\bY}{{\bar Y}}
\newcommand{\bU}{{\bar U}}
\newcommand{\bV}{{\bar V}}

\newcommand{\tX}{{\tilde X}}
\newcommand{\tY}{{\tilde Y}}
\newcommand{\tU}{{\tilde U}}
\newcommand{\tV}{{\tilde V}}

\newcommand{\mkv}{-\!\!\!\!\minuso\!\!\!\!-}

\newcommand\numberthis{\addtocounter{equation}{1}\tag{\theequation}}

\arxiv{1604.1292}

\startlocaldefs
\endlocaldefs

\begin{document}

\begin{frontmatter}

\title{Collaborative Distributed Hypothesis Testing}
\runtitle{Collaborative Distributed HT}

%
\thankstext{t1}{Large Systems and Networks Group (LANEAS),
	CentraleSup\'elec-CNRS-Universit\'e Paris-Sud,
	Gif-sur-Yvette, France.} 
\thankstext{t2}{Laboratoire des Signaux et Syst\`emes (L2S),
	CentraleSup\'elec-CNRS-Universit\'e Paris-Sud,
	Gif-sur-Yvette, France.}
\begin{aug}
\author{\fnms{Gil} \snm{Katz}\thanksref{t1}\ead[label=e1]{gil.katz@CentraleSupelec.fr}},
\author{\fnms{Pablo} \snm{Piantanida}\thanksref{t2}\ead[label=e2]{pablo.piantanida@CentraleSupelec.fr}},
\and
\author{\fnms{M\'erouane} \snm{Debbah}\thanksref{t3}\ead[label=e3]{merouane.debbah@CentraleSupelec.fr}\thanksref{t1}}
\thankstext{t3}{Mathematical and Algorithmic Sciences Lab,
	Huawei France R\&D,
	Paris, France}
\address{\printead{e1,e2,e3}}
\runauthor{G. Katz et al.}
\affiliation{CentraleSup\'elec and Huawei France}
\end{aug}

\runauthor{G. Katz, P. Piantanida and M. Debbah}

\begin{abstract}
A collaborative distributed binary decision problem is considered. Two statisticians are required to declare the correct probability measure of two jointly distributed memoryless process, denoted by $X^n=(X_1,\dots,X_n)$ and $Y^n=(Y_1,\dots,Y_n)$, out of two possible probability measures on finite alphabets, namely $P_{XY}$ and $P_{\bar{X}\bar{Y}}$. The marginal samples given by  $X^n$ and $Y^n$ are assumed to be available at different locations. The statisticians are allowed to exchange limited amount of data over multiple rounds of interactions, which differs from previous work that deals mainly with unidirectional communication. A single round of interaction is considered before the result is generalized to any finite number of communication rounds. A feasibility result is shown, guaranteeing the feasibility of an error exponent for general hypotheses, through information-theoretic methods. The special case of testing against independence is revisited as being an instance of this result for which also an unfeasibility result is proven.  A second special case is studied where zero-rate communication is imposed  (data exchanges grow sub-exponentially with $n$) for which it is shown that interaction does not  improve asymptotic performance.



\end{abstract}



\end{frontmatter}

\section{Introduction}
The field of hypothesis testing (HT) is comprised of different problems, in which the goal is to determine the probability measure (PM) of one or more random variables (RVs), based on a number of available observations. Considering binary HT problems, it is assumed that this choice is made out of two possible hypotheses, denoted the null hypothesis $H_0$ and the alternative hypothesis $H_1$. In this setting, two error events may occur: An error of Type I, with probability $\alpha_n$ (dependent on the number of observations $n$), occurs when the alternative hypothesis $H_1$ is declared while $H_0$ is true. Conversely, an error of Type II with probability $\beta_n$, occurs when $H_0$ is declared despite $H_1$ being true. Often, for fixed $0<\epsilon <1$, the goal is to find the optimal error exponent: 
 \begin{equation}
 E(\epsilon) \coloneqq \liminf\limits_{n \to \infty} -\frac{1}{n}\log \beta_n(\epsilon) \ ,
 \end{equation}
 for a constrained error probability of Type I: $\alpha_n \leq \epsilon$.

Let $\left\{X_i\right\}_{i=1}^\infty$ be an independent and identically distributed (i.i.d) process, commonly refereed to as a \emph{memoryless process}, taking values in a countably finite alphabet $\mathcal{X}$ equipped with probability measures  $P_0$ or $P_1$ defined on the measurable space $(\mathcal{X}, \mathcal{B}_\mathcal{X})$, where $\mathcal{B}_\mathcal{X}=2^{\mathcal{X}}$. Denote  $X^n=(X_1,\dots,X_n)$ the finite block of the process following the product measures $P_0^n$ or $P_1^n$ on $(\mathcal{X}^n, \mathcal{B}_{\mathcal{X}^n})$. Let us denote by $\mathcal{P}(\mathcal{X})$ the family of probability measures in $(\mathcal{X}, \mathcal{B}_\mathcal{X})$, where for every $\mu \in \mathcal{P}(\mathcal{X})$, $f_\mu(x) \coloneqq \frac{d \mu}{d\lambda}(x) = \mu(  \left\{ x\right\})$ is a short-hand for its probability mass function (pmf). The optimal error exponent for the Type II error probability of the binary HT problem is well-known and given by \emph{Stein's Lemma} (see e.g., \cite{Lehmann-2005,Cover-Thomas-1991}) to be: 
  \begin{equation}
 E(\epsilon)= \cD(P_0|| P_1) \ ,\ \, \forall\,0<\epsilon<1
  \end{equation}
  where $P_0$ and $P_1$ are the probability measures implied by hypotheses $H_0$ and $H_1$, respectively, and $\cD(\cdot||\cdot)$ is the \emph{Kullback-Leiber divergence}  satisfying $P_0 \ll P_1$. The optimal exponential rate of decay of the error probability of Type II does not depend on the specific constraint over the error probability of Type I. This property is referred to as \emph{strong unfeasibility}.
  
  \begin{figure}[bt]
  	\centering
  	\includegraphics[trim=3cm 9cm 7.5cm 4cm, clip=true,width=0.6\textwidth]{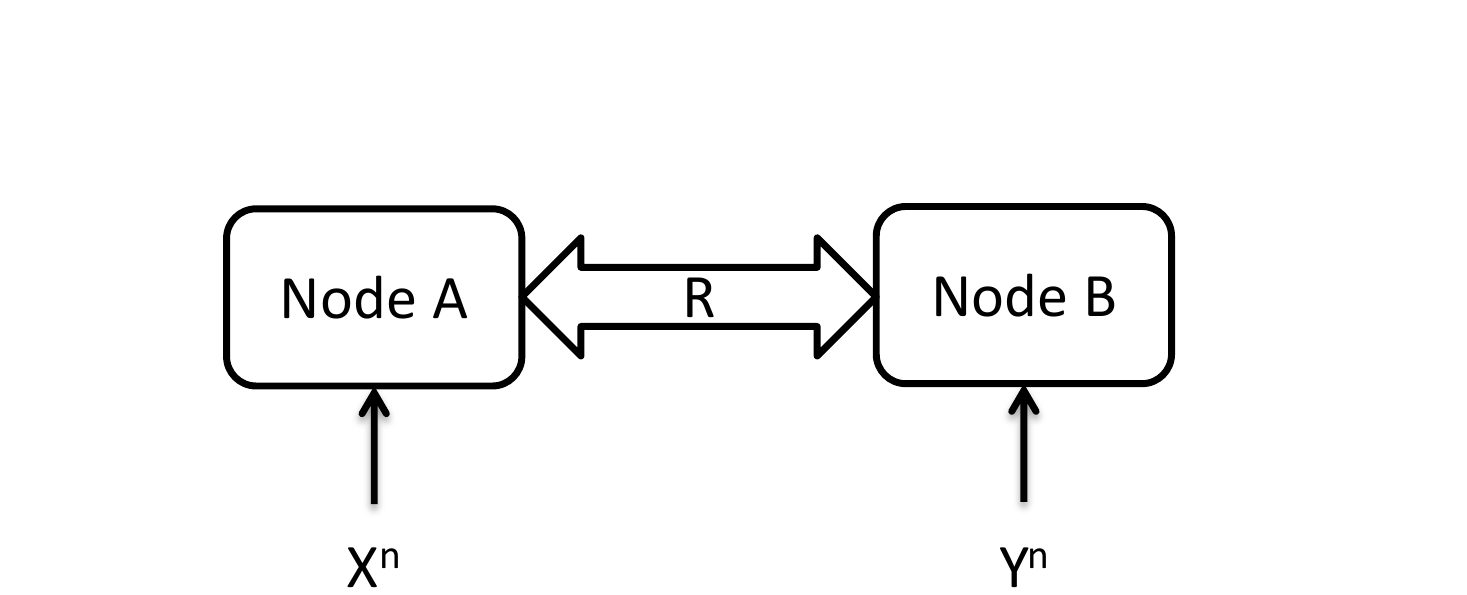}
  	\caption{Collaborative Distributed Hypothesis Testing model.}
  	\label{fig:Model}
  \end{figure}
  
In many scenarios, the realizations of different parts of a random process are available at different physical locations (with different statisticians) in the system (see Fig.~\ref{fig:Model}). Assuming that exchanging  data between the statisticians is possible but costly, a new question arises --for a given constraint over the total amount of data exchange between the nodes, what is the optimal error exponent to the error probability  of Type II, under a fixed constraint over the error probability of Type I? In this paper, we compose together two stories. One is from statistics concerning binary HT originating in the works of Wald~\cite{Wald-1945,Wald-Wolfowitz-48}. The other story is from information theory concerning the case of \emph{unidirectional} data exchanges where only one statistician can share information with the other one due to~\cite{Ahlswede-Csiszar-1986,Han-1987}. We focus on \emph{bidirectional collaborative} binary HT problem. It is assumed that the available resources for interaction can be divided between the statisticians  in any way that would benefit performance, and that without loss of generality no importance is given to the location at which the decision is made -- as the decision can always be transmitted with sub-exponential resources. First, we concentrate on a special case where only one ``round of interaction'' (only a query and its reply) is allowed between the statisticians, i.e., a decision is made after each statistician communicates one statistics, which will be commonly referred to as a message. This scenario was first studied in~\cite{Xiang-Kim-2012} for a special case called \emph{testing against independence}. While the scenario studied in this paper borrows ideas from~\cite{Xiang-Kim-2012}, the mathematical tools are fundamentally different since these rely on the \emph{method of types}~\cite{Csiszar-1998}, as it was the case to deal with general hypothesis in~\cite{Han-1987}. We then extend our result for any finite number of interaction rounds, before showing that this new result for general hypotheses implies the special case of testing against independence,  for which optimality is proven via an \emph{unfeasibility} property.

The remainder of this paper is organized as follows. We finish this introduction with a short summary of related results, before presenting the considered statistical model in Section~\ref{Sec:Model}. In Section~\ref{Sec:MainResult}, we present and prove our first result, being a feasible error exponent for the case of general hypotheses and interactive exchanges, under the assumption of a single communication round. Section~\ref{Sec:ManyRounds} extends this result to any \emph{finite number} of interaction rounds. In Section~\ref{Sec:AgainstIndependence}, we revisit the special case of testing against independence and show that the known exponent for this case is indeed feasible through our general exponent result. Then,  we show an  unfeasibility  property (thus proving optimality, at least in a ``weak'' sense) for the case of a single communication round.  In Section~\ref{Sec:ZeroRate}, we give the optimal error exponent when communication is constrained to be of \emph{zero rate}, meaning that the sizes of the codebooks grows sub-exponentially with the number of observations $n$. 

\subsection{Summary of related works}
Some of the first contributions on binary HT are due to Wald~\cite{Wald-1945,Wald-Wolfowitz-48} where an optimal course of action is given by which a sequential probability ration test (SPRT) is used. It was shown that the expected number of observations required to reach a conclusion is lower than by any other approach, when a similar constraint over the probabilities of error is enforced. Stein's Lemma takes an information-theoretic form since by considering the limit where the number of observations $n \to \infty$, it is shown that the optimal  error exponent for the error probability of Type II, under any fixed constraint over the error probability of Type I, is given by the KL divergence. Later~\cite{Chiyonobu-2001} proves an important property by which when $\alpha_n\equiv \exp(-nc) \to 0$ as $n\to \infty$, then $\beta_n \to 0$ or $\beta_n \to 1$, exponentially depending on the rate of decay $c>0$.

Among the first works that started enforcing constraints on the basic HT problem, which are independent from the statistical nature of the data, are references~\cite{Cover-69,Hellman-Cover-70}. The single-variable HT is considered, and the enforced constraint is related to the \emph{memory} of the system, rather than to communication between different locations. It is assumed that a realistic system cannot hold a large number of observations for future use, and thus at each step a function must be used that would best encapsulate the ``knowledge'' gained from the new observation, combined with the compressed representation of previous observations. This problem was then revisited in~\cite{Yakowitz-1974,Bucklew-Ney-91}, which are motivated by new scenarios in which memory efficiency is an important aspect, such as satellite communication systems. \cite{Bucklew-Ney-91} focuses on the case where both probabilities of error simultaneously decay to zero.

Distributed HT with communication constraints was the focus of the seminal works~\cite{Ahlswede-Csiszar-1986,Han-1987}. Both of them investigated binary decisions in presence of a helper, i.e., unidirectional communication, and propose a feasible  error exponent for $\beta_n$ while enforcing a strict constraint over $\alpha_n$. Although  both of these approaches achieve optimality for the case of testing against independence, where it is assumed that under the alternative hypothesis $H_1$ the samples from $(X,Y)$ are independent with the same marginal measures implied by $H_0$, optimal results for the case of general hypotheses remain allusive until this day. Improving these results by using further randomization of the codebooks, referred to ``random binning'', was first briefly suggested in~\cite{Shimokawa-han-amari-1994} and analyzed thoroughly in~\cite{IT-2016}. In~\cite{Ahlswede-Burnashev-90} a similar scenario is considered for parameter estimation with unidirectional communication. This is a generalization of the binary HT problem where the mean square-error loss was considered instead of exponential decay of the error probability.

A special case referred to as  HT under ``complete data compression'' was studied in~\cite{Han-1987}. In this case, it is assumed that node $A$ is allowed to communicate with node $B$ by sending only one bit of information. A feasible scheme was proposed and its optimality proved. The much broader scenario, by which codebooks are allowed to grow with $n$, but not exponentially fast, was studied in~\cite{Shalaby-Papamarcou-1992}. Interestingly, it was shown that this scenario does not offer any advantage, with relation to complete data compression. This setting, referred to as zero-rate communication, was recently revisited in~\cite{Zhao-Lai-2015} where both $\alpha_n$ and $\beta_n$ are required to decrease exponentially with $n$.

Interactive communication was considered for the problem of distributed binary HT within the framework of testing against independence in~\cite{Xiang-Kim-2012}. In the present paper, we further study this problem in the framework of general hypotheses, as well as revisit the special case of testing against independence via a strong unfeasibility  proof. Other works in recent years evolve the problem of HT in many different directions. Two interesting examples are~\cite{Naghshvar-Javidi-2013} (see references therein), which assumes a tighter control by the statistician throughout the process, allowing him to choose and evaluate the testing procedure through past information, and~\cite{Nussbaum-Szkola-2009} which investigates HT in the framework of quantum statistical models.

\section{Statistical Model and Preliminaries}\label{Sec:Model}

\newcounter{InEquations}
\stepcounter{InEquations}
\newcounter{OutEquations}
\stepcounter{OutEquations}

\subsection{Notation} 
We use upper-case letters to denote random variables (RVs) and lower-case letters to denote realizations of RVs. Vectors are denoted by boldface letters, with their length as a superscript, emitted when it is clear from the context. Sets, including alphabets of RVs, are denoted by calligraphic letters. Throughout this paper we assume all RVs have an alphabet of finite cardinality. $P_X \in \mathcal{P}(\cX)$ denotes a probability measure (PM) for the RV $X\in\mathcal{P}(\mathcal{X})$ defined on the measurable space $(\mathcal{X}, \mathcal{B}_\mathcal{X})$,  that belongs to the set of all possible PMs over $\cX$; $X \mkv Y \mkv Z$ denotes that $X$, $Y$ and $Z$ form a Markov chain. We shall use tools from information theory. Notations generally comply with the ones introduced in~\cite{Csiszar-1998}. Thus, for a RV $X$, distributed by $X \sim P_X(x)$, the  \emph{entropy} is defined to be $H(X) = H(P) \coloneqq -\sum\limits_{x \in \cX} P_X(x)\log P_X(x)$. Similarly,  the  \emph{conditional  entropy}:
$$
H(Y|X) = H(V|P) \coloneqq -\sum_{x \in \cX}\sum_{y \in \cX}  P_X(x) V(y|x)  \log V(y|x) 
$$ 
for a stochastic mapping  $V :\mathcal{X}\mapsto  \mathcal{P}(\mathcal{Y})$. The \emph{conditional Kullback-Leiber (KL) divergence} between two stochastic mappings  $P_{Y|X} :\mathcal{X}\mapsto  \mathcal{P}(\mathcal{Y})$ and $Q_{Y|X}: \mathcal{X}\mapsto \mathcal{P}(\mathcal{Y})$, is:
\begin{equation}
\cD(P_{Y|X}\|Q_{Y|X}| P_{X}) \coloneqq \sum\limits_{x \in \cX}\sum\limits_{y \in \cY}  P_X(x) P_{Y|X} (y|x) \log \frac{P_{Y|X} (y|x)}{Q_{Y|X} (y|x)} \ ,
\end{equation}
satisfying that $P_{Y|X} \ll Q_{Y|X}$ \emph{a.e.} wrt $P_X$.  For any two RVs, $X$ and $Y$, whose measure is controlled by $XY \sim P_{XY}(x,y) = P_X(x)P_{Y|X}(y|x)$, the following is defined to be the \emph{mutual information} between them: $I(X;Y) \coloneqq  \cD(P_{XY}\|P_{X}P_{Y}) $. Given a vector  $\vct{x}=(x_1,\dots,x_n)\in \mathcal{X}^n$, let $N(a|\vct{x})$ be the \emph{counting measure}, i.e., the number of times the letter $a \in \cX$ appears in the vector $X$. The \emph{type} of the vector $\vct{x}$, denoted by $Q_{\vct{x}}$,  is defined through its \emph{empirical measure}: $Q_{\vct{x}}(a) = n^{-1}N(a|\vct{x})$ with $a \in \mathcal{X}$. $\cP_n(\cX)$ denotes the set of all possible types (or empirical measures) of length $n$ over $\cX$.  We use type variables of the form $X^{(n)}\in\mathcal{P}_n(\mathcal{X})$ to denote a RV with a probability measure identical to the empirical measure induced by $\vct{x}$. The set of all vectors $\vct{x}$ that share this type is denoted by $\cT(Q_{\vct{x}}) = \cT_{[Q_{\vct{x}}]}$. Main definitions of $\delta$-\emph{typical sets} and some of their properties, are given in Appendix~\ref{Apen:typicality}. All exponents and logarithms are assumed to be of base $2$.

\subsection{Statistical model and problem statement} 
In a system comprising two statisticians, as depicted in Fig.~\ref{fig:Model}, each of them is assumed to observe the i.i.d. realizations of one random variable. Let $X^nY^n=(X_1,Y_1),\dots,$ $(X_n,Y_n)$  be independent random variables  in $(\cX^n\times\cY^n,\mathcal{B}_{\cX^n\times\cY^n})$ that are jointly distributed in one of two ways, denoted by hypothesis $0$ and $1$, with probability measures as follows: 
\begin{equation}
\left\{\begin{aligned}
& H_0: \quad P_{XY}(x,y)\ , \forall\,(x,y)\in\mathcal{X}\times\mathcal{Y} \ ,\\
& H_1: \quad P_{\bX\bY}(x,y)\ , \forall\,(x,y)\in\mathcal{X}\times\mathcal{Y} \ .
\end{aligned}\right.
\end{equation}
Communication between the two statisticians is assumed to be done in rounds, with node $A$ starting the interaction. These   interactions are limited, however, by a total (exponential) rate $R$ bits per symbol. That is, if each of the nodes sees $n$ realizations, the total amount of bits allowed to exchange data between the nodes before the decision is made is $\exp(nR)$. The data exchange is assumed to be \emph{perfect}, meaning that within the rate limit no errors are introduced by the communication. It is assumed that the total rate can be distributed by the two statisticians in any way that is beneficial to performance. Moreover, we assume that it does not matter \emph{where} the decision is finally made, as its transmission can be done at no cost.

As is the case in the standard centralized HT problem, we consider two error events. An error of the Type I, with probability $\alpha_n$, occurs when $H_1$ is declared despite $H_0$ being true, while an error event of Type II, with probability $\beta_n$, is the opposite error event. The goal is to find the exponential rate: $ -\frac{1}{n}\log \beta_n$  ($n$ being the number of samples) s.t. $\beta_n\to0$ as $n\to\infty$, while fixed constraints are enforced on $\alpha_n$ and the total exchange rate $R$. 

\begin{definition}[K-round collaborative HT]\label{def-main}
A $K$-round decision code for the two node collaborative hypothesis testing system, when each of the statisticians  is allowed to observe $X^n$ and $Y^n$ realizations of $X$ and $Y$, respectively, is defined by a sequence of encoders and a decision mapping:
	\begin{align}
	 f_{[k]}&: \mathcal{X}^n \times \prod_{i=1}^{k-1} \{1,\dots,| g_{[i]}|\}   \longrightarrow \{1,\dots,|  f_{[k]} |\}  \ , \ k=[1:K]\\
	    g_{[k]}& : \mathcal{Y}^n \times \prod_{i=1}^{k}\{1,\dots,| f_{[i]}| \}\longrightarrow\{1,\dots, |  g_{[k]} |\}  \ ,\ k=[1:K]\\
	\phi&: \mathcal{X}^n \times\prod_{i=1}^{K} \{1,\dots,| g_{[i]}| \} \longrightarrow \{0,1\} \ ,
	\end{align}
where $f_{[k]}$ and $g_{[k]}$ are encoder mappings with image sizes satisfying $\log | f_{[i]}|\equiv \mathcal{O}(n)$ and $\log| g_{[i]}|\equiv \mathcal{O}(n) $, respectively, while $\phi$ is the decision mapping. The corresponding Type I and II error probabilities are given by
\begin{align}
\alpha_n(R\,|K)& \coloneqq  \Pr\left[ \phi\big(X^n,g_{[1:K]}\big)=1 \, |  \,X^nY^n\sim P_{XY} \right] \ , \\
\beta_n(R\,|K)& \coloneqq  \Pr\left[ \phi\big(X^n,g_{[1:K]}\big)=0 \, |  \,X^nY^n\sim P_{\bX\bY} \right]  \ .
\end{align}

An  exponent $E$ to the error probability of Type II, constrained  to an error probability of Type I to be below $\epsilon>0$ and a total exchange rate $R$,  is said to be feasible, if for any $\varepsilon >0$ there exists a code satisfying:	
\begin{align}
	 -\frac{1}{n}\log \beta_n(R,\epsilon\,|K)	  &\geq E - \varepsilon   \ ,\\
	\frac{1}{n} \sum\limits_{k=1}^K \log \left(| g_{[k]}||  f_{[k]} | \right) &\leq R +\varepsilon \  , \ \alpha_n(R\,|K)\leq \epsilon\ ,
	\end{align}
	provided that $n$ is large enough. The supremum of all feasible exponents for given $(R,\epsilon)$ is defined to be the optimal error exponent. 
\end{definition}

\section{Collaborative Hypothesis Testing with One Round} \label{Sec:MainResult}
In this section, we present and prove a \emph{feasible error exponent} $ -\frac{1}{n} \log \beta_n(R,\epsilon\,|K=1)$  to the error probability of Type II, under any fixed constraint $\epsilon>0$ on the error probability of Type I for a total exchange rate $R$. Here, we only consider one round of exchange whereby each of the nodes exchanges one statistics (or message) before a decision is made. The extension to the case with multiple exchanging rounds is relegated to the next section.

\begin{Prop}[Sufficient conditions for one round of interaction]\label{Prop:Main} 
Let $\mathscr{S}(R)\subset \mathcal{P}(\mathcal{U}\times\mathcal{V})$ and $\mathscr{L}(U,V) \subset \mathcal{P}(\mathcal{U}\times\mathcal{V}\times\mathcal{X}\times\mathcal{Y})$ denote the sets of probability measures defined in terms of  corresponding RVs: 
	\begin{align}
	\mathscr{S}(R) &\coloneqq \big\{ UV :  \ I(U;X) + I(V;Y|U) \leq R\, \\
	&\quad \quad  U \mkv X \mkv Y\,, \, V \mkv (U,Y) \mkv X\ ,\,\, | \mathcal{U}|, | \mathcal{V}|<+\infty\big\} \ ,\nonumber\\
	&\mathscr{L}(U,V)  \coloneqq\big \{\tilde{U}\tilde{V}\tilde{X}\tilde{Y}\,: \, P_{\tilde{U}\tilde{V}\tilde{X}} =   P_{UVX}\, ,\, P_{\tilde{U}\tilde{V}\tilde{Y}} = P_{UVY} \big\} \ .
	\end{align}
A feasible  error exponent to the error probability of Type II, when the  total exchange rate is $R$ (bits per sample), is given by 
	\begin{align}
	&\lim_{\epsilon \to 0}\,\liminf_{n\to\infty} -\frac{1}{n}  \log \beta_n(R,\epsilon\,|K=1) \geq \\
	&\quad\quad\quad\quad\quad\quad\quad\max\limits_{UV \in \mathscr{S}(R)}\, \min\limits_{\tU\tV\tX\tY \in \mathscr{L}(U,V)} \cD\big(P_{\tU\tV\tX\tY}||P_{\bU\bV\bX\bY}\big) \ .\nonumber
	\end{align} 
\end{Prop}
\begin{proof}
We start by describing the random construction of codebooks, as well as encoding and decision functions. By analyzing the asymptotic properties of such decision systems, we aim at implying a \emph{feasibility (existence) result} of interactive functions and decision regions that satisfy, for any given $\epsilon,\varepsilon>0$, the following inequalities: 
\begin{align}
	 &\frac{1}{n}\log \left(|f_{[1]}| |g_{[1]}|\right) \leq I(U;X) + I(V;Y|U) + \varepsilon \ ,\ 	\alpha_n(R\,|K=1)  \leq \epsilon\ ,\\
	 &-\frac{1}{n}\log \beta_n(R,\epsilon\,|K=1) \geq  \min\limits_{\tU\tV\tX\tY \in \mathscr{L}(U,V)} \cD\big(P_{\tU\tV\tX\tY}||P_{\bU\bV\bX\bY}\big) - \varepsilon \ , 
	\end{align}
provided that $n$ is large enough and for any given pair of random variables $(U,V)\in\mathscr{S}(R)$, where $|f_{[1]}|$ and $|g_{[1]}|$ denote the number of codewords in the codebooks\footnote{Note that \emph{feasibility}  is defined in the information-theoretic sense which implies the  \emph{random existence} of interactive and decision functions with desired properties.} used for interaction.\vspace{1mm}

	\emph{Codebook generation}. Without loss of generality, we assume that node $A$ is the first to communicate. Fix a conditional probability $P_{UV|XY}(u,v|x,y) = P_{U|X}(u|x)P_{V|UY}(v|u,y)$ that attains the maximum in Proposition~\ref{Prop:Main}. Let 
\begin{equation*}
P_U(u) \equiv  \sum\limits_{x \in \cX} P_{U|X}(u|x)P_X(x) \ ,\  P_{V|U}(v|u) \equiv  \sum\limits_{y \in \cY} P_{V|UY}(v|u,y)P_{Y}(y). 
\end{equation*}
For this choice of RVs, set the rates $(R_U,R_V)$ to be
\begin{equation*}
I(U;X) + \epsilon(\delta)\coloneqq  R_U \ , \  I(V;Y|U) + \epsilon(\delta^\prime) \coloneqq R_V 
\end{equation*}
with $\epsilon(\delta) \to 0$ as $\delta \to 0$. By the definition of the set $\mathscr{S}(R)$, it is clear that $R_U + R_V \leq R + \epsilon(\delta) + \epsilon(\delta^\prime)$. Randomly and independently draw $2^{nR_U}$ sequences $\vct{u}=(u_1,\dots,u_n)$, each according to $\prod_{i=1}^n P_U(u_i)$. Index these sequences by $m_U \in [1:M_U \coloneqq 2^{nR_U}]$ to form the random codebook 
$\mathcal{C}_{\vct{u}}\coloneqq \big\{ \vct{u}(m_U): \, m_U \in [1:M_U]\big\}$. As a second step, for each word $\vct{u}\in \mathcal{C}_{\vct{u}}$, build a codebook $\mathcal{C}_{\vct{v}}(m_U)$ by randomly and independently drawing $2^{nR_V}$ sequences $\vct{v}$, each according to  $\prod_{i=1}^n P_{V|U}(v_i|u_i(m_U))$. Index these sequences by $m_V \in [1:M_V \coloneqq  2^{nR_V}]$ to form the collection of codebooks $\mathcal{C}_{\vct{v}}(m_U) \coloneqq \big\{ \vct{v}(m_U,m_V): \, m_V \in [1:M_V]\big\}$ for $m_U \in [1:M_U]$.\vspace{1mm}
 
\emph{Encoding and decision mappings}. Given a sequence $\vct{x}$, node $A$ searches in the codebook $\mathcal{C}_{\vct{u}}$ for an index $m_U$ such that $(\vct{u}(m_U),\vct{x}) \in \cT_{[UX]_\delta}^n$ (note that this notation denotes the $\delta$-typical set with relation to the probability measure implied by $H_0$). If no such index is found, node $A$ declares $H_1$. If more than one sequence is found, node $A$ chooses one at random. Node $A$ then communicates  the chosen index $m_U$ to node $B$, using a portion $R_U$ bits of the available exchange rate. Upon receiving the index $m_U$, node $B$ checks if $(\vct{u}(m_U),\vct{y}) \in \cT_{[UY]_{\delta'}}^n $. If not, node $B$ declares $H_1$. If the received sequence $\vct{u}$ and $\vct{y}$ (the observed sequence at node $B$) are jointly typical, node $B$ looks in the specific codebook $\mathcal{C}_{\vct{v}}(m_U)$, for an index $m_V$ such that $\big(\vct{u}(m_U),\vct{v}(m_U,m_V),\vct{y}\big) \in \cT_{[UVY]_{\delta'}}^n$. If such an index is not found, node $B$ declares $H_1$. If node $B$ finds more than one such index, it chooses one of them at random. Node $B$ then transmits the chosen index $m_V$ to node $A$. Upon reception of the index $m_V$, node $A$ checks if $\big(\vct{u}(m_U),\vct{v}(m_U,m_V),\vct{x}\big) \in \cT_{[UVX]_{\delta''}}^n $. If so, it declares $H_0$ and otherwise, it declares $H_1$. The relation between $\delta,\delta'$ and $\delta''$ can be deducted from Lemma~\ref{Lemma:JointConditionalTypicality}. It is, however, important to emphasize that $\delta'(\delta) \to 0$ as $\delta \to 0$, and $\delta''(\delta') \to 0$ as $\delta' \to 0$ with $n\to \infty$.\vspace{1mm}
	
\emph{Analysis of $\alpha_n$ (Type I)}. The analysis of $\alpha_n$ is identical to the one proposed in \cite{Xiang-Kim-2012}, for the case of testing against independence. We give here a short summary of the analysis available in \cite{Xiang-Kim-2012}. Assuming that the measure that controls $X$ and $Y$ is $P_{XY}$, and denoting the chosen indices at nodes $A$ and $B$ by $m_U$ and $m_V$ respectively, the error probability of the Type I can be expressed as follows 
	\begin{equation}\label{eq:AnalysisAlpha}
	\alpha_n \equiv  \Pr(\cE_1 \cup \cE_2 \cup \cE_3) \leq \Pr(\cE_1) + \Pr (\cE_1^c \cap \cE_2) + \Pr (\cE_1^c \cap \cE_2^c \cap \cE_3) \ ,
	\end{equation}
	where $\cE_1,\cE_2$ and $\cE_3$ represent the following error events:
\begin{align}
\cE_1 &\equiv \bigl\{ (\vct{U}(m_U),\vct{X}) \notin \cT^n_{[UX]\delta}\, \forall \,m_U \in [1:M_U]\bigr\} \ ,\\
\cE_2 &\equiv \bigl\{ (\vct{V}(m_U,m_V),\vct{U}(m_U),\vct{Y}) \notin \cT^n_{[VUY]_{\delta'}} \,\forall \, m_V \in [1:M_V] \\&  \qquad\text{ and the specific $m_U$ selected at node $A$}\bigr\} \ ,\nonumber\\
\cE_3 &\equiv \bigl\{ (\vct{V}(m_U,m_V),\vct{U}(m_U),\vct{X}) \notin \cT^n_{[VUX]_{\delta''}}, \\&  \qquad\text{ for the specific $m_U$ and $m_V$ previously chosen}\bigr\} \ .\nonumber
\end{align}
Analyzing each of the probabilities in \eqref{eq:AnalysisAlpha} separately, $\Pr(\cE_1)\to0$ as $n \to \infty$ by the \emph{covering lemma}~\cite{ElGamal-Kim-2011}, provided that $R_U \geq I(U;X) + \epsilon(\delta)$, with $\epsilon(\delta) \to 0$ as $\delta \to 0$. $\Pr (\cE_1^c \cap \cE_2) \to 0$ when $n \to \infty$ by the \emph{conditional typicality lemma} \cite{ElGamal-Kim-2011}, in addition to the covering lemma, provided that $R_V \geq I(V;Y|U) + \epsilon(\delta')$. Finally, the third term in \eqref{eq:AnalysisAlpha} can be shown to tend to zero through the use of the Markov lemma (see Lemma~\ref{Lemma:markov}), as well as Lemma~\ref{Lemma:ProbabilityByType} and Lemma~\ref{Lemma:JointConditionalTypicality} in Appendix~\ref{Apen:typicality}. Thus, as all three components tend to zero with large $n$, we may conclude that $\alpha_n \leq \epsilon$ for any constraint $0<\epsilon <1$ and $n$ large enough.\vspace{1mm}
	
\emph{Analysis of $\beta_n$ (Type II)}. The error probability of Type II is defined by
	\begin{equation}
	\beta_n(R,\epsilon\,|K=1) \equiv \Pr \big(\textrm{decide $H_0$}|XY \sim P_{\bX\bY}\big) \ .
	\end{equation}
	Thus, we assume that $P_{\bX\bY}$ controls the measure of the observed RVs throughout this analysis. We use similar methods to what was done in \cite{Han-1987}, although we choose to work with random codebooks. The influence of this choice is on the analysis of $\alpha_n$ only, as seen above, and not on $\beta_n$.
	
	For a given pair of sequences $(\vct{x},\vct{y})$ with type variables $X^{(n)}Y^{(n)}\in\mathcal{P}_n(\mathcal{X}\times\mathcal{Y})$, we count all possible events that lead to an error. We notice first, that given a pair of vectors $(\vct{x},\vct{y})\in \mathcal{X}^n\times\mathcal{Y}^n$ the probability that these vectors will be the result of $n$ i.i.d. draws, according to the measure implied by $H_1$, is given by Lemma~\ref{Lemma:ProbabilityByType} to be:
	\begin{equation}
	\Pr \{\bar{X}^n\bar{Y}^n = (\vct{x},\vct{y})\} = \exp \left[-n\left(H(X^{(n)}Y^{(n)}) + \cD(X^{(n)}Y^{(n)}||\bar{X}\bar{Y})\right)\right] \ ,
	\end{equation}  
	where $X^{(n)}Y^{(n)}\in\mathcal{P}_n(\mathcal{X}\times\mathcal{Y})$ are the type variables of the realizations $(\vct{x},\vct{y})$ (see Appendix~\ref{Apen:typicality}). For each pair of codewords $\vct{u}_i\in \mathcal{C}_{\vct{u}}$ and  $\vct{v}_{ij}\in \mathcal{C}_{\vct{v}}(i)$, we define the set: 
	\begin{equation}\label{Eq:S}
	\mathcal{S}_{ij}(\vct{x}) \coloneqq \{\vct{u}_i\}\times\{\vct{v}_{ij}\}\times\mathcal{G}_{ij}\times\{\vct{x}\} \ ,
	\end{equation}
	where $\mathcal{G}_{ij}\subseteq \mathcal{Y}^n$ is the set of all vectors $\vct{y}$ that, given the received message $\vct{u}_i$, will result in the message $\vct{v}_{ij}$ being transmitted back to node $A$. Denoting by $K_{ij}(\vct{x})$ the number of elements $(\vct{u}_i,\vct{v}_{ij},\vct{x},\vct{y}) \in\mathcal{S}_{ij}(\vct{x})$ whose type variables coincide with $U^{(n)}V^{(n)}X^{(n)}Y^{(n)}$, we have by Lemma~\ref{Lemma:SizeTypicalSet} that: 
	\begin{equation}\label{Eq:Kij}
	K_{ij}(\vct{x}) \leq \exp \left[nH(Y^{(n)}|U^{(n)}V^{(n)}X^{(n)})\right] \ .
	\end{equation}
	Let  $K(U^{(n)}V^{(n)}X^{(n)}Y^{(n)})$ denote the number of all elements:
	$$
	(\vct{u},\vct{v},\vct{x},\vct{y}) \in \mathscr{S}_n \coloneqq \bigcup_{i=1}^{M_U} \bigcup_{j=1}^{M_V}\, \bigcup_{\vct{x} \in \mathcal{T}_{[X|\vct{u}_i\vct{v}_{ij}]_{\delta''}}^n}\mathcal{S}_{ij}(\vct{x})
	$$ 
	that have type variable $U^{(n)}V^{(n)}X^{(n)}Y^{(n)}\in\mathcal{P}_n(\mathcal{U}\times\mathcal{V}\times\mathcal{X}\times\mathcal{Y})$, then
	\begin{equation}\label{Eq:K}
	\begin{aligned}
	K( U^{(n)} V^{(n)}X^{(n)}Y^{(n)}) &\leq \sum\limits_{i=1}^{M_U}\sum\limits_{j=1}^{M_V} \exp \left[nH(Y^{(n)}|U^{(n)}V^{(n)}X^{(n)})\right]  \! \big|\mathcal{T}_{[X|\vct{u}_i\vct{v}_{i,j}]_{\delta^{\prime\prime}}}^n\big| \\
	&\leq \exp \left[n\left(H(Y^{(n)}|U^{(n)}V^{(n)}X^{(n)}) \right.\right.\\ & \qquad\quad+ I(U;X) + I(V;Y|U) + H(X|UV) + \mu_n\Big)\Big]  \ ,
	\end{aligned}
	\end{equation}
	where $M_U$ and $M_V$ are the sizes of the codebooks $\mathcal{C}_{\vct{u}}$ and $\mathcal{C}_{\vct{v}}(\cdot)$. The first and second additional terms in the final expression come from the size of the codebooks and the third is a bound over the size of the delta-typical set (see Lemma~\ref{Lemma:SizeDeltaTypicalSet}). The resulting sequence $\mu_n$ is a function of $\delta,\delta',\delta''$ that complies with $\mu_n \to 0$ as $n \to \infty$. The error probability of  Type II satisfies: 
	\begin{equation}\label{eq:BoundOnBetaSum}
	\begin{aligned}
	\beta_n(R,\epsilon\,|K=1) \leq \sum\limits_{U^{(n)}V^{(n)}X^{(n)}Y^{(n)}\in{\mathscr{S}_n}} \!\!\!\!\exp \left[-n\left(k(U^{(n)}V^{(n)}X^{(n)}Y^{(n)}) -\mu_n\right)\right] \ ,
	\end{aligned}
	\end{equation}
	where the function $k(U^{(n)}V^{(n)}X^{(n)}Y^{(n)})$ is defined by 
	\begin{equation}\label{Eq:k}
	\begin{aligned}
	k(U^{(n)}V^{(n)}X^{(n)}Y^{(n)}) &\coloneqq  H(X^{(n)}Y^{(n)}) + \cD(X^{(n)}Y^{(n)}|| \bar{X}\bar{Y})\\ &\quad  - H(Y^{(n)}|U^{(n)}V^{(n)}X^{(n)})-H(X|UV) \\ &\quad - I(U;X) - I(V;Y|U) \ .
	\end{aligned}
	\end{equation}
	
	We deliberately made an abuse of notation in  \eqref{eq:BoundOnBetaSum} to indicate that the sum  is taken over all possible type-variables $U^{(n)}V^{(n)}X^{(n)}Y^{(n)}\in\mathcal{P}_n(\mathcal{U}\times\mathcal{V}\times\mathcal{X}\times\mathcal{Y})$ formed by empirical probability measures from elements $(\vct{u},\vct{v},\vct{x},\vct{y})\in\mathscr{S}_n$.
	
	From the construction of $\mathscr{S}_n$, it is clear that if $(\vct{u},\vct{v},\vct{x},\vct{y}) \in \mathscr{S}_n$, then at least $(\vct{u},\vct{v},\vct{x}) \in \cT_{[UVX]_{\delta''}}^n$ and $(\vct{u},\vct{v},\vct{y}) \in \cT_{[UVY]_{\delta'}}^n$. Thus, the summation in \eqref{eq:BoundOnBetaSum} is only over all types satisfying:
	\begin{equation}\label{eq:TypicalForBeta}
	\begin{aligned}
	&|P_{U^{(n)}V^{(n)}X^{(n)}}(u,v,x) - P_{UVX}(u,v,x)| \leq \delta''  \ ,\\
	&|P_{U^{(n)}V^{(n)}Y^{(n)}}(u,v,y) - P_{UVY}(u,v,y)| \leq \delta'  \ ,
	\end{aligned}
	\end{equation}
	for all $(u,v,x)\in\textrm{supp}(P_{UVX})$ and $(u,v,y)\in\textrm{supp}(P_{UVY})$. 
	In addition, it follows by Lemma~\ref{Lemma:TypeCounting} from the total number of types of length $n$ that:
	\begin{equation}\label{eq:BoundOnBetaMax}
	\begin{aligned}
	&\beta_n(R,\epsilon\,|K=1) \leq (n+1)^{|\cU||\cV||\cX||\cY|} \\ &\quad\quad\times\max\limits_{U^{(n)}V^{(n)}X^{(n)}Y^{(n)}\in\mathscr{S}_n} \exp \left[-n\left(k(U^{(n)}V^{(n)}X^{(n)}Y^{(n)}) -\mu_n\right)\right] \ .
	\end{aligned}
	\end{equation}
	By \eqref{eq:TypicalForBeta} and the continuity of the entropy function as well as the KL divergence, we can conclude that
	\begin{align}
	k(U^{(n)}V^{(n)}X^{(n)}Y^{(n)}) &= H(\tilde{X}\tilde{Y}) + \cD(\tilde{X}\tilde{Y}|| \bar{X}\bar{Y}) - H(\tilde{Y}|\tilde{U}\tilde{V}\tilde{X})\\ & -H(\tilde{X}|\tilde{U}\tilde{V}) - I(\tilde{U};\tilde{X}) - I(\tilde{V};\tilde{Y}|\tilde{U}) + \mu'_n\ ,\nonumber
	\end{align}
	with $\tU\tV\tX\tY \in \mathscr{L}(U,V)$
	and $\mu'_n \to 0$ when $n \to \infty$. We can further simplify the expression of  $k(U^{(n)}V^{(n)}X^{(n)}Y^{(n)})$ by observing that:
	\allowdisplaybreaks
	\begin{align*}\label{Eq:AnalysisOfK}
	&\numberthis k(U^{(n)} V^{(n)}X^{(n)}Y^{(n)}) =
			 H(\tilde{X}\tilde{Y}) + \cD(\tilde{X}\tilde{Y}|| \bar{X}\bar{Y}) - H(\tilde{Y}|\tilde{U}\tilde{V}\tilde{X}) \\ &\qquad\qquad\qquad\qquad\qquad -H(\tilde{X}|\tilde{U}\tilde{V}) - I(\tilde{U};\tilde{X})- I(\tilde{V};\tilde{Y}|\tilde{U}) + \mu'_n\\
	& = H(\tilde{X}\tilde{Y}) + \cD(\tilde{X}\tilde{Y}|| \bar{X}\bar{Y}) - H(\tilde{X}\tilde{Y}|\tilde{U}\tilde{V}) - I(\tilde{U};\tilde{X}) -I(\tilde{V};\tilde{Y}|\tilde{U})+ \mu'_n \\
	& = I(\tilde{X}\tilde{Y};\tilde{U}\tilde{V}) + \cD(\tilde{X}\tilde{Y}|| \bar{X}\bar{Y})- I(\tilde{U};\tilde{X}) -I(\tilde{V};\tilde{Y}|\tilde{U}) + \mu'_n\\
	&= I(\tilde{X}\tilde{Y};\tilde{U}) + I(\tilde{X}\tilde{Y};\tilde{V}|\tilde{U})+ \cD(\tilde{X}\tilde{Y}|| \bar{X}\bar{Y})- I(\tilde{U};\tilde{X})  -I(\tilde{V};\tilde{Y}|\tilde{U})+ \mu'_n\\
	&\overset{(\alph{InEquations}\stepcounter{InEquations})}{=} \cD(\tilde{U}\tilde{X}\tilde{Y}||\bar{U}\bar{X}\bar{Y}) + I(\tilde{X}\tilde{Y};\tilde{V}|\tilde{U}) - I(\tilde{Y};\tilde{V}|\tilde{U})+ \mu'_n\\
	&\overset{}{=}  \cD(\tilde{U}\tilde{X}\tilde{Y}||\bar{U}\bar{X}\bar{Y}) + I(\tilde{X};\tilde{V}|\tilde{U}\tilde{Y}) + \mu'_n\ ,
	\end{align*}
where equality $(\alph{OutEquations}\stepcounter{OutEquations})$ stems from the identity~\cite{Han-1987}: 
	\begin{align}
	I(\tX\tY;\tU) + \cD(\tX\tY||\bX\bY) - I(\tU;\tX) &= I(\tU;\tY|\tX) + \cD(\tX\tY||\bX\bY) \\
	&=\cD(\tU\tX\tY||\bU\bX\bY)\ ,\nonumber
	\end{align}
which holds the case on unidirectional communication. Note that the following Markov chain: $X \mkv (U,Y) \mkv V$ holds under both hypotheses (i.e., the same chain can be written with a bar over all variables), \emph{but not} for the auxiliary RVs, marked with a tilde. 

Finally, we conclude our development of  $k(U^{(n)}V^{(n)}X^{(n)}Y^{(n)})$ as follows:
	\begin{align}\label{Eq:Prop1EndOfProof}
	& k(U^{(n)} V^{(n)}X^{(n)}Y^{(n)}) =
	\cD(\tilde{U}\tilde{X}\tilde{Y}||\bar{U}\bar{X}\bar{Y}) + I(\tilde{X};\tilde{V}|\tilde{U}\tilde{Y})+ \mu'_n\\& = \sum_{\forall(u, v, x, y)} P_{\tilde{U}\tilde{V}\tilde{X}\tilde{Y}}(u,v,x,y)\times \nonumber\\ 
	&\qquad \times\log \left(\frac{P_{\tilde{U}\tilde{X}\tilde{Y}}(u,x,y)}{P_{\bar{U}\bar{X}\bar{Y}}(u,x,y)}\frac{P_{\tilde{X}\tilde{V}|\tilde{U}\tilde{Y}}(x,v|u,y)}{P_{\tilde{X}|\tilde{U}\tilde{Y}}(x|u,y)P_{\tilde{V}|\tilde{U}\tilde{Y}}(v|u,y)}\right)+ \mu'_n\nonumber\\
	&\overset{(\alph{InEquations}\stepcounter{InEquations})}{=} \sum_{\forall(u, v, x, y)} P_{\tilde{U}\tilde{V}\tilde{X}\tilde{Y}}(u,v,x,y)\log \left(\frac{P_{\tilde{U}\tilde{V}\tilde{X}\tilde{Y}}(u,v,x,y)}{P_{\bar{U}\bar{X}\bar{Y}}(u,x,y)P_{\bar{V}|\bar{U}\bar{Y}}(v|u,y)}\right)+ \mu'_n\nonumber\\
	&= \sum_{\forall(u, v, x, y)} P_{\tilde{U}\tilde{V}\tilde{X}\tilde{Y}}(u,v,x,y)\log \left(\frac{P_{\tilde{U}\tilde{V}\tilde{X}\tilde{Y}}(u,v,x,y)}{P_{\bar{U}\bar{V}\bar{X}\bar{Y}}(u,v,x,y)}\right)+ \mu'_n\nonumber\\
	& = \cD(\tilde{U}\tilde{V}\tilde{X}\tilde{Y}||\bar{U}\bar{V}\bar{X}\bar{Y}) + \mu'_n \ ,\nonumber
	\end{align}
	where the sums are over the $\textrm{supp}(P_{\tilde{U}\tilde{V}\tilde{X}\tilde{Y}})$; and $(\alph{OutEquations}\stepcounter{OutEquations})$ is due to  the definition of the set $\mathscr{L}(U,V)$ that implies $P_{\tV|\tU\tY}(v|u,y) = P_{V|UY}(v|u,y)$. In addition, as coding (at each side) is performed before a decision is made, it is clear it is done in the same way under both hypotheses. Thus, while $P_{UVY}(u,v,y) \neq P_{\bU\bV\bY}(u,v,y)$, it is true that $P_{\bV|\bU\bY}(v|u,y) = P_{V|UY}(v|u,y) = P_{\tV|\tU\tY}(v|u,y)$. As $\mu_n,\mu'_n$ are arbitrarily small, as a function of the choices of $\delta$ and $\delta'$ provided that $n$ is large enough, this concludes the proof of Proposition~\ref{Prop:Main}.
\end{proof}

\section{Collaborative Hypothesis Testing with Multiple Rounds} \label{Sec:ManyRounds}
We now allow the statisticians  to exchange data over an arbitrary but \emph{finite} number of exchange rounds, and investigate the extension of Proposition~\ref{Prop:Main} to this more general case. The corresponding result is stated below.

\begin{Prop}[Sufficient conditions for $K$-rounds of interaction] \label{Prop:ManyRounds}Let $\mathscr{S}(R)$ and $\mathscr{L}\big(U_{[1:K]},V_{[1:K]}\big)$ denote the sets of probability measures defined in terms of corresponding RVs: 
	\begin{align}
	&\mathscr{S}(R) \coloneqq  \Bigl\{U_{[1:K]}V_{[1:K]} : R \geq \sum\limits_{k=1}^K \big[I(X;U_{[k]}|U_{[1:k-1]}V_{[1:k-1]})\\ &\quad \quad \quad \quad \quad + I(Y;V_{[k]}|U_{[1:k-1]}V_{[1:k-2]}) \big]\ ,\nonumber\\ 
	&  U_{[k]} \mkv \big(X,U_{[1:k-1]},V_{[1:k-1]}\big) \mkv Y\ , \  |\mathcal{U}_{[k]} |<+\infty \ ,\nonumber \\
	&V_{[k]} \mkv \big(Y,U_{[1:k]},V_{[1:k-1]}\big) \mkv X\ , \  |\mathcal{V}_{[k]} |<+\infty \ , \forall\, k\in[1:K] \Bigr\} \ ,\nonumber\\
	&\mathscr{L}\big(U_{[1:K]},V_{[1:K]}\big)
	\coloneqq \Bigl\{ \tU_{[1:K]}\tV_{[1:K]}\tX\tY:\\
	  & \quad\, P_{\tU_{[1:K]}\tV_{[1:K]}\tX} = P_{U_{[1:K]}V_{[1:K]}X} \ , \,\, P_{\tU_{[1:K]}\tV_{[1:K]}\tY} = P_{U_{[1:K]}V_{[1:K]}Y}\Bigr\} \ ,\nonumber
	\end{align}
	where $U_{[1:k]} \coloneqq   (U_{[1]},\dots,U_{[k]})$ and $V_{[1:k]} \coloneqq   (V_{[1]},\dots,V_{[k]})$ represent the exchanged data between nodes $A$ and $B$  until round $k$. A feasible  error exponent to the error probability of Type II, when the  total (over $K$-rounds)  exchange rate is $R$ (bits per sample), is given by 
	\begin{align}
	&\lim_{\epsilon \to 0}\,\liminf_{n\to\infty} -\frac{1}{n}  \log \beta_n(R,\epsilon|K) \geq \\
	&\quad\quad\quad\quad\quad\max\limits_{\mathscr{S}(R)}\min\limits_{\mathscr{L}\big(U_{[1:K]},V_{[1:K]}\big)} \cD\Big(P_{\tU_{[1:K]}\tV_{[1:K]}\tX\tY}\big |\big|P_{\bU_{[1:K]}\bV_{[1:K]}\bX\bY}\Big) \ .\nonumber
	\end{align} 
\end{Prop}
This proposition is very clearly an extension of Proposition~\ref{Prop:Main} to allow multiple rounds of interaction. The implication of this result is as follows. Given a limited budget of rate $R$ for data exchange, which the nodes can divide as they choose into any finite number of $K$ exchange rounds, the gain of interaction attained through the different characteristics of the underlying Markov process between the  RVs comes at no cost in terms of the form of the expression for the error exponent.

\begin{proof}[Proof of Proposition~\ref{Prop:ManyRounds}]
The proof of this proposition is very similar to the one presented above for Proposition~\ref{Prop:Main}. Codebook construction, as well as encoding and decision mappings remain similar. At each round, a codebook is built based on any possible combination of the previous messages. Given previous messages, each node chooses a message in the relevant codebook and communicates its index to the other statistician. The process continues until a message cannot be found, which is jointly typical with all previous messages as well as the observed sequence, in which case $H_1$ is declared. Otherwise, until the end of round $K$ in which case $H_0$ is declared, provided that all the messages are jointly typical with the observed sequence. We next provide a sketch of the proof to this simple extension. 
	
	
The analysis of $\alpha_n$ applies similarly to the previous case, as long as a finite number of rounds is considered. Regarding the analysis of $\beta_n$, the following important changes are needed: 
	\begin{itemize}
		\item The set $\mathcal{S}_{\textbf{i}\textbf{j}}(\vct{x})$ is now defined by using all exchanged messages:
		\begin{equation}
		\mathcal{S}_{\textbf{i}\textbf{j}}(\vct{x}) \coloneqq \{\vct{u}_{[1],i_1}\}\times\{\vct{v}_{[1],i_1j_1}\}\times\cdots\times\{\vct{u}_{[K],i_K}\}\times\{\vct{v}_{[K],i_Kj_K}\} \times \mathcal{G}_{\textbf{i}\textbf{j}} \times \{\vct{x}\} \ ,
		\end{equation}
		where $(\textbf{i},\textbf{j})\coloneqq(i_1,j_1),\dots,(i_K,j_K)$ and $\vct{u}_{[k],i_k}$ is the $i_k$-th message in the codebook $\mathcal{C}_{\vct{u}_{[k]}}$,  similarly for the other random variables.
		\item Similarly, $\mathscr{S}_n$ is now defined by the union over the codewords of \emph{all} auxiliary RVs.
		\item The bound over $K_{\textbf{i}\textbf{j}}$ (analogues to expression \eqref{Eq:Kij} before) writes:
		\begin{equation}
		K_{\textbf{i}\textbf{j}}(\vct{x}) \leq \exp \left[nH\big(Y^{(n)}|U_{[1:K]}^{(n)}V_{[1:K]}^{(n)}X^{(n)}\big) \right] \ .
		\end{equation}
		\item Finally, $K\big(U_{[1:K]}^{(n)}V_{[1:K]}^{(n)}X^{(n)}Y^{(n)}\big)$, i.e., see \eqref{Eq:K},  is now calculated through the summation over the codebooks of all messages, considering the cardinality of the conditional set: $\big|\cT^n_{[X|\vct{u}_{[1:K],\textbf{i}}\vct{v}_{[1:K],\textbf{i}\textbf{j}}]_{\delta}}\big|$.
		\item As more steps are performed, each of which requires encoding, we also need to define new $\delta$'s for each of these steps. We refrain from this for the sake of readability, as all of these $\delta$'s go to $0$ together, as was seen in the case of a single round.
	\end{itemize}
	Considering these differences, after $k$ rounds of interactions, $k\big(U_{[1:k]}^{(n)}V_{[1:k]}^{(n)}\big)$  can be shown to be equal to (e.g. see \eqref{Eq:k}):
	\begin{align}
&k\big(U_{[1:k]}^{(n)}V_{[1:k]}^{(n)}\big)   = \cD\big(P_{\tU_{[1:k-1]}\tV_{[1:k-1]} \tX\tY}||P_{\bU_{[1:k-1]}\bV_{[1:k-1]} \bX\bY}\big)\\ 
	&  \quad+ I(\tY;\tU_{[k]}|\tU_{[1:k-1]}\tV_{[1:k-1]} \tX) + I(\tX;\tV_{[k]}|\tU_{[1:k]}\tV_{[1:k-1]}\tY) + \mu'_n\ .\nonumber
	\end{align}
	By continuing in the same manner as in \eqref{Eq:Prop1EndOfProof}, we show:
\allowdisplaybreaks
\begin{align*}
&k \big(U_{[1:k]}^{(n)} V_{[1:k]}^{(n)}\big) - \mu'_n = \sum_{\forall} P_{\tU_{[1:k-1]}\tV_{[1:k-1]}\tX\tY} \log \frac{P_{\tU_{[1:k-1]}\tV_{[1:k-1]}\tX\tY}}{P_{\bU_{[1:k-1]}\bV_{[1:k-1]}\bX\bY}}\\ 
&+ \sum_{\forall} P_{\tU_{[1:k]}\tV_{[1:k-1]}\tX\tY} \log \frac{P_{\tU_{[k]}\tY|\tU_{[1:k-1]}\tV_{[1:k-1]}\tX}}{P_{\tU_{[k]}|\tU_{[1:k-1]}\tV_{[1:k-1]}\tX}P_{\tY|\tU_{[1:k-1]}\tV_{[1:k-1]}\tX}}\\
&+\sum_{\forall} P_{\tU_{[1:k]}\tV_{[1:k]}\tX\tY} \log \frac{P_{\tV_{[k]}\tX|\tU_{[1:k]}\tV_{[1:k-1]}\tY}}{P_{\tV_{[k]}|\tU_{[1:k]}\tV_{[1:k-1]}\tY}P_{\tX|\tU_{[1:k]}\tV_{[1:k-1]}\tY}}\\
&= \sum_{\forall} P_{\tU_{[1:k]}\tV_{[1:k]} \tX\tY}  \log \left[ \frac{P_{\tU_{[1:k]}\tV_{[1:k]}\tX\tY}}{P_{\bU_{[1:k-1]}\bV_{[1:k-1]}\bX\bY}P_{\tU_{[k]}|\tU_{[1:k-1]}\tV_{[1:k-1]}\tX}P_{\tV_{[k]}|\tU_{[1:k]}\tV_{[1:k-1]}\tY}} \right]\\
&\overset{(\alph{InEquations}\stepcounter{InEquations})}{=} \sum_{\forall} P_{\tU_{[1:k]}\tV_{[1:k]}\tX\tY} \log \left[ \frac{P_{\tU_{[1:k]}\tV_{[1:k]}\tX\tY}}{P_{\bU_{[1:k-1]}\bV_{[1:k-1]}\bX\bY}P_{\bU_{[k]}|\bU_{[1:k-1]}\bV_{[1:k-1]}\bX}P_{\bV_{[k]}|\bU_{[1:k]}\bV_{[1:k-1]}\bY}} \right]\\	
&= \sum_{\forall} P_{\tU_{[1:k]}\tV_{[1:k]}\tX\tY} \log \left[ \frac{P_{\tU_{[1:k]}\tV_{[1:k]}\tX\tY}}{P_{\bU_{[1:k]}\bV_{[1:k]}\bX\bY}} \right]\\
&= \cD\big(P_{\tU_{[1:k]}\tV_{[1:k]}\tX\tY}||P_{\bU_{[1:k]}\bV_{[1:k]}\bX\bY}\big) \ , 
\end{align*}
where all sums are over all the alphabets of the relevant RVs. Here, $(\alph{OutEquations}\stepcounter{OutEquations})$, much like in the case of single-round exchange above, is due to the definition of the set $\mathscr{L}(U_{[1:k]},V_{[1:k]})$ and to the fact that encoding occurs without knowledge of the PM controlling the RVs, and thus behaves the same under each of the hypotheses. Thus, 
$$
P_{\tU_{[k]}|\tU_{[1:k-1]}\tV_{[1:k-1]}\tX} = P_{U_{[k]}|U_{[1:k-1]}V_{[1:k-1]}X} = P_{\bU_{[k]}|\bU_{[1:k-1]}\bV_{[1:k-1]}\bX},
$$ 
and similarly for the messages $V_{[k]}$ at node $B$. Pursuing this until round $K$, the proposition is proved.
\end{proof}

\begin{remark}
	For reasons of brevity and clarity, we chose in this paper to concentrate on scenarios where the interactions begins and ends at node $A$. However, it is easy to see that this does not necessarily need to be the case. The process could start or end at node $B$, implying that the final round of exchange  is in fact only half of a round, without any significant changes to the theory or our proofs.
\end{remark}

\section{Collaborative Testing Against Independence} \label{Sec:AgainstIndependence}
We now concentrate on the special problem of testing against independence, where it is assumed that under $H_1$ the $n$ observed samples of the RVs $(X,Y)$ defined on $(\mathcal{X}\times\mathcal{Y}, \mathcal{B}_{\mathcal{X}\times\mathcal{Y}})$ 
are distributed according to a product measure:
\begin{equation}\label{eq-testing-independence}
\begin{cases}
H_0: & P_{XY}(x,y)\ , \forall\,(x,y)\in\mathcal{X}\times\mathcal{Y} \ ,\\
H_1: & P_{\bX\bY}(x,y) = P_X(x)P_Y(y)\ , \forall\,(x,y)\in\mathcal{X}\times\mathcal{Y}  \ ,
\end{cases}
\end{equation}
where $P_X(x)$ and $P_Y(y)$ are the marginal probability measures implied by $P_{XY}(x,y)$. Testing against independence was first studied, for a unidirectional communication link~\cite{Han-1987} (see also \cite{Ahlswede-Csiszar-1986}). It was shown that the \emph{optimal} rate of exponential decay to the error probability of Type II is:
\begin{align}
\liminf_{n\to\infty}-\frac{1}{n} \log \beta_n(R,\epsilon|K=1/2) & =\\
\max_{\begin{array}{l}
P_{U|X}:\mathcal{X}\mapsto \mathcal{P}(\mathcal{U}) \\  \text{ s.t. $I(U;X)  \leq R$}
\end{array}
} &I(U;Y) \ ,\ \, \forall\,0<\epsilon<1\ ,\nonumber
\end{align}
where  $R$ is the available exchange rate from node $A$ to node $B$. Note that much like the case of centralized HT, the \emph{optimal} error exponent does not depend on $\epsilon$ and thus a \emph{strong unfeasibility} (converse) result holds.

Testing against independence in a cooperative scenario was first studied in~\cite{Xiang-Kim-2012}, for the case of a single round of interaction. It was shown that a feasible error exponent to the error probability of Type II  is given by
\begin{equation}\label{Eq:TestingAgainstIndependenceB}
\lim_{\epsilon \to 0}\,\liminf_{n\to\infty}-\frac{1}{n} \log \beta_n(R,\epsilon|K=1)  \geq E(R)
\end{equation}
subject  to a total available exchange rate $R$, where:
\begin{equation}\label{Eq:TestingAgainstIndependence}
E(R)\coloneqq \!\!\! \!\!\! \!\!\! \!\!\!\!\!\! \!\!\! \max_{\begin{array}{c}
P_{U|X}:\mathcal{X}\mapsto  \mathcal{P}(\mathcal{U})\\ P_{V|UY}:\mathcal{U}\times\mathcal{Y}\mapsto \mathcal{P}( \mathcal{V})\\
 \text{ s.t. $I(U;X)  + I(V;Y|U)\leq R$}
\end{array}
}\!\!\! \!\!\!  \!\!\! \!\!\! \!\!\! \!\!\!\big[I(U;Y) + I(V;X|U)\big]\ .
\end{equation}
While the proof of feasibility inspired the approach taken in Proposition~\ref{Prop:Main} for general hypotheses, unfortunately, the auxiliary RVs identified in the \emph{weak} unfeasibility  proof in~\cite{Xiang-Kim-2012} do not match the required Markov chains to lead to a feasible exponent (the reader may refer to~\cite{Vega-Piantanida-Hero-2015, Kaspi-1985} for further details). 

In this section, we revisit the problem of characterizing the reverse inequality in~\eqref{Eq:TestingAgainstIndependenceB}. We 
prove a \emph{weak unfeasibility} result, determining necessary and sufficient conditions to the optimality of the error exponent~\eqref{Eq:TestingAgainstIndependence} satisfying $\alpha_n \leq \epsilon$ for \emph{any} $0<\epsilon <1$ (i.e., we prove that the exponent in ~\eqref{Eq:TestingAgainstIndependence} is optimal in the case where we constrain $\alpha_n$ to go to $0$ with $n$). We first show that Proposition~\ref{Prop:Main} implies the feasibility part, i.e., inequality~\eqref{Eq:TestingAgainstIndependenceB}, and then follow with a new proof for the  unfeasibility (for $\epsilon$ arbitrarily small) of any higher exponent.

\begin{theorem}[Necessary and sufficient conditions for testing against independence with $K=1$]\label{Theo:AgainstIndependenceOptimality}
The \emph{optimal}  error exponent  to the error probability of Type II for testing against independence is given by 
\begin{equation}
\lim_{\epsilon \to 0}\,\liminf_{n\to\infty}-\frac{1}{n} \log \beta_n(R,\epsilon|K=1)  \coloneqq  E(R)\ , \ \textrm{ $\forall$ $0<\epsilon <1$} \ ,
\end{equation}
where $E(R)$ is defined in~\eqref{Eq:TestingAgainstIndependence}, and $R$ denotes the available rate of interaction between the statisticians and $\epsilon$ is the error probability of Type I. 
\end{theorem}

\begin{remark}
In  a similar manner to Theorem~\ref{Theo:AgainstIndependenceOptimality}, a \emph{feasible} error exponent to the error probability of Type II with $K$ rounds is given by
\begin{align}\label{eq:ManyRounds}
&\lim_{\epsilon \to 0}\,\liminf_{n\to\infty}-\frac{1}{n} \log \beta_n(R,\epsilon|K) \geq \\ 
 \max\limits_{U_{[1:K]}V_{[1:K]} \in \mathscr{S}(R)}
& \sum\limits_{k=1}^K \left[I\big(U_{[k]};Y|U_{[1:k-1]}V_{[1:k-1]}\big)\right. + \left. I\big(V_{[k]};X|U_{[1:k]}V_{[1:k-1]}\big)\right]\ . \nonumber 
\end{align}
The proof of the feasibility of \eqref{eq:ManyRounds} follows largely the same path as the one for the feasibility part provided below for Theorem~\ref{Theo:AgainstIndependenceOptimality}. However, for $K>1$ our unfeasibility proof does not hold and this feasible exponent result may not longer be optimal.
\end{remark}

\subsection{Proof of Theorem~\ref{Theo:AgainstIndependenceOptimality}}
We first enunciate and prove some preliminary results from which the proof of Theorem~\ref{Theo:AgainstIndependenceOptimality} will easily follow. 

\begin{lemma}[Multi-letter representation for testing against independence with $K=1$ \cite{Xiang-Kim-2012}]\label{Prop:WeakConverse}
The error exponent to the error probability of Type II  for testing against independence with one round satisfies:  
\begin{align}\label{eq:ExpressionToExpend}
\lim_{\epsilon \to 0}\,\liminf_{n\to\infty} -\frac{1}{n} & \log \beta_n(R,\epsilon\,|K=1) \leq \frac{1}{n} \big[I(I_A;Y^n) + I(I_B;X^n|I_A) \big]\ ,\\
R& \geq \frac{1}{n} \big[I(I_A;X^n) + I(I_B;Y^n|I_A) \big]\ ,
\end{align}
where $I_A\coloneqq f_1(X^n)$ and $I_B\coloneqq g_1\big(f_1(X^n),Y^n\big)$ for any mappings $(f_1,g_1)$, as given in Defintion~\ref{def-main}.
\end{lemma}
\begin{proof}
The proof follows~\cite{Ahlswede-Csiszar-1986,Xiang-Kim-2012} and is given in Appendix~\ref{Apen:WeakConverse}.
\end{proof}


\begin{proof}[Proof of Theorem~\ref{Theo:AgainstIndependenceOptimality}]
We start showing the feasibility, followed by a proof of the unfeasibility part.

\subsubsection*{Feasibility}
In order to show the feasibility  to the exponent~\eqref{Eq:TestingAgainstIndependence} through the general result stated in  Proposition~\ref{Prop:Main}, it is convenient to use the form of the last expression in~\eqref{Eq:AnalysisOfK}:
	\begin{align} \label{Eq:StartAchievableAgainstIndependence}
\liminf_{n\to\infty}&-\frac{1}{n} \log \beta_n(R,\epsilon|K=1)  \geq \\
&\max\limits_{UV\in \mathscr{S}(R)} \min\limits_{\tU\tV\tX\tY \in \mathscr{L}(U,V)} \left[\cD(P_{\tilde{U}\tilde{X}\tilde{Y}}||P_{\bar{U}\bar{X}\bar{Y}}) + I(\tilde{X};\tilde{V}|\tilde{U}\tilde{Y})\right] \ .\nonumber
	\end{align}
We analyze each of these components separately:
\begin{equation}
\begin{aligned}
\cD(P_{\tilde{U}\tilde{X}\tilde{Y}}||P_{\bar{U}\bar{X}\bar{Y}}) &\overset{(\alph{InEquations}\stepcounter{InEquations})}{=} \cD(P_{\tilde{U}\tilde{Y}}||P_{\bar{U}\bar{Y}}) + \cD(P_{\tX|\tU\tY}||P_{\bX|\bU\bY}|P_{\tU\tY})\\
& \overset{(\alph{InEquations}\stepcounter{InEquations})}{=} I(U;Y) + \cD(P_{\tX|\tU\tY}||P_{\bX|\bU}|P_{\tU\tY})\\
&  \, \,{=}  \, \,I(U;Y) + \cD(P_{\tX|\tU\tY}||P_{\tX|\tU}|P_{\tU\tY}) + \cD(P_{\tX|\tU}||P_{\bX|\bU}|P_{\tU})\\
&\overset{(\alph{InEquations}\stepcounter{InEquations})}{\geq}I(U;Y) + \cD(P_{\tX|\tU\tY}||P_{\tX|\tU}|P_{\tU\tY}) \ ,
\end{aligned}
\end{equation}
where  $(\alph{OutEquations}\stepcounter{OutEquations})$ is due to the chain rule and $\cD(P_{\tX|\tU\tY}||P_{\bX|\bU\bY}|P_{\tU\tY})$ 
is the conditional KL-divergence; $(\alph{OutEquations}\stepcounter{OutEquations})$ stems from the assumption of testing against independence, as well as the Markov chain $\bU \mkv \bX \mkv \bY$ and the fact that $P_{\tU\tY} = P_{UY}$; and $(\alph{OutEquations}\stepcounter{OutEquations})$ is due to the fact that the KL-divergence is non-negative. To conclude the analysis, we note that:
	\begin{align}\label{eq:PlusForAgainstIndependence}
	&\cD(P_{\tilde{U}\tilde{X}\tilde{Y}}||P_{\bar{U}\bar{X}\bar{Y}}) \geq \\
	&\quad I(U;Y) + \sum\limits_{(u,x,y) \in \cU\times\cX \times \cY}  P_{\tilde{U}\tilde{X}\tilde{Y}}(u,x,y)\log\left(\frac{P_{\tilde{X}|\tilde{U}\tilde{Y}}(x|u,y)}{P_{\tilde{X}|\tilde{U}}(x|u)}\right)\nonumber \\
	&=I(U;Y) + \sum\limits_{(u,x,y) \in \cU\times\cX \times \cY}  P_{\tilde{U}\tilde{X}\tilde{Y}}(u,x,y)\log\left(\frac{P_{\tilde{X}\tilde{Y}|\tilde{U}}(x,y|u)}{P_{\tilde{X}|\tilde{U}}(x|u)P_{\tilde{Y}|\tilde{U}}(y|u)}\right)\nonumber\\
	&= I(U;Y) + I(\tilde{X};\tilde{Y}|\tilde{U}) \ .\nonumber
	\end{align}
	As for the second term in \eqref{Eq:StartAchievableAgainstIndependence}, we express it as follows:
	\begin{equation}\label{eq:MinusForAgianstIndependence}
	I(\tilde{V};\tilde{X}|\tilde{U}\tilde{Y}) = I(\tilde{V}\tilde{Y};\tilde{X}|\tilde{U}) - I(\tilde{X};\tilde{Y}|\tilde{U}) \geq I(\tilde{V};\tilde{X}|\tilde{U}) - I(\tilde{X};\tilde{Y}|\tilde{U}) \ .
	\end{equation}
	This allows us to conclude through \eqref{Eq:StartAchievableAgainstIndependence} that
	\begin{equation}
	\begin{aligned}
	\liminf_{n\to\infty}-\frac{1}{n} \log \beta_n(R,\epsilon\,|K=1) & \geq  \max\limits_{UV\in \mathscr{S}(R)} \min\limits_{\tU\tV\tX\tY \in \mathscr{L}(U,V)} \left[I(U;Y) + I(\tilde{V};\tilde{X}|\tilde{U})\right]\\
	& = \max\limits_{UV\in \mathscr{S}(R)} \left[I(U;Y) + I(V;X|U) \right] \ ,
	\end{aligned}
	\end{equation}
	which completes the proof of feasibility  through Proposition~\ref{Prop:Main}.

	\subsubsection*{Weak unfeasibility}
We are now ready to complete the proof of weak unfeasibility (converse) to Theorem~\ref{Theo:AgainstIndependenceOptimality}. From Lemma~\ref{Prop:WeakConverse}, it follows that:
\begin{align}	\label{eq:ExpressionToExpendC}
\lim_{\epsilon \to 0}\,\liminf_{n\to\infty} & -\frac{1}{n}  \log \beta_n(R,\epsilon\,|K=1) \\
& \leq \limsup_{n\to\infty}\, \frac1n [I(I_A;Y^n) + I(I_B;X^n|I_A)]\coloneqq \limsup_{n\to\infty} \Delta_n \ ,\nonumber
\end{align}	
where $I_A$ is the message sent from node A  while $I_B$ is its reply from node B.  In order to derive a single-letter expression, we expand \eqref{eq:ExpressionToExpendC} as follows:
\begin{equation}\label{eq:BoundE}
	\begin{aligned}
&\Delta_n \overset{(\alph{InEquations}\stepcounter{InEquations})}{=}  \frac1n\sum_{i=1}^n \left[I(I_A;Y_i|Y_{i+1}^n) + I(I_B;X_i|I_AX^{i-1})\right]\\
	&\overset{(\alph{InEquations}\stepcounter{InEquations})}{=} \frac1n \sum_{i=1}^n \left[I(I_AY_{i+1}^n;Y_i) + I(I_BY_{i+1}^n;X_i|I_AX^{i-1}) - I(Y_{i+1}^n;X_i|I_AI_BX^{i-1})\right]\\
	&= \frac1n\sum_{i=1}^n \left[I(I_AX^{i-1}Y_{i+1}^n;Y_i) - I(X^{i-1};Y_i|I_AY_{i+1}^n) + I(Y_{i+1}^n;X_i|I_AX^{i-1}) \right.\\ &\qquad\quad \left.+ I(I_B;X_i|I_AX^{i-1}Y_{i+1}^n) -  I(Y_{i+1}^n;X_i|I_AI_BX^{i-1})\right]\\
	&\overset{(\alph{InEquations}\stepcounter{InEquations})}{=}\frac1n \sum_{i=1}^n \left[I(\hat{U}_i;Y_i) + I(V_i;X_i|\hat{U}_i) - I(Y_{i+1}^n;X_i|I_AI_BX^{i-1})\right] \ ,
	\end{aligned}
	\end{equation}
	where $X^{i}$ denotes the first $i$ samples and $X_{i}^n=(X_i,\dots,X_n)$; $(\alph{OutEquations}\stepcounter{OutEquations})$ stems from the chain rule and $(\alph{OutEquations}\stepcounter{OutEquations})$ from the assumed i.i.d. nature of the sources. In $(\alph{OutEquations}\stepcounter{OutEquations})$, the following identity is used~\cite{Csiszar-Korner-2011}:
	\begin{equation} \label{eq:Csiszar-Korner}
	\sum\limits_{i=1}^n I(\vct{A}^{i-1};B_i|C,\vct{B}_{i+1}^n) = \sum\limits_{i=1}^n I(\vct{B}_{i+1}^n;A_i|C,\vct{A}^{i+1}) \ ,
	\end{equation} 
	where $C$ can be arbitrarily dependent to the vectors $\vct{A}$ and $\vct{B}$, as long as it does not change with $i$, 
	and the following auxiliary RVs are defined on measurable spaces $(\mathcal{U}_i\times\mathcal{V}_i, \mathcal{B}_{\mathcal{U}_i\times\mathcal{V}_i})$ by setting:
	\begin{equation}\label{eq:ChosenRandomVariables}
	\hat{U}_i \coloneqq  (I_A,X^{i-1},Y_{i+1}^n) \ \ \textrm{ and } \ \ V_i \coloneqq I_B \ , \ \forall \, i=[1:n]\ .
	\end{equation}
	It is important to emphasize that the required Markov chains in~\eqref{Eq:TestingAgainstIndependence} are verified for each $i=[1:n]$ (see Appendix~\ref{Apen:MarkovChains}). Let $Q$ be a RV uniformly distributed over $[1:n]$, then:
	\begin{equation}\label{eq:TimeSharing}
	\begin{aligned}
	\Delta_n &\leq I(\hat{U}_Q;Y_Q|Q) + I(V_Q;X_Q|\hat{U}_Q,Q)-  \frac1n \sum_{i=1}^n I(Y_{i+1}^n;X_i|I_AI_BX^{i-1})  \\
	&=  I(U;Y) + I(V;X|U) - T \ ,
	\end{aligned}
	\end{equation}
	where $U \coloneqq (\hat{U}_Q,Q)$. We now bound the required rate, from the size of the mappings, we have
	\begin{equation}
	nR \geq  I(I_A;X^n) + I(I_B;Y^nI_A) \geq I(I_A;X^n) + I(I_B;Y^n|I_A) \ .
	\end{equation}
	For convenience,  we analyze each of these terms separately: 
	\begin{equation}
	\begin{aligned}
	I(I_A;X^n) & \overset{(\alph{InEquations}\stepcounter{InEquations})}{=} \sum_{i=1}^n I(I_AX^{i-1};X_i) \\
	& = \sum_{i=1}^n \left[I(I_AX^{i-1}Y_{i+1^n};X_i) - I(Y_{i+1}^n;X_i|I_AX^{i-1})\right] \ ,
	\end{aligned}
	\end{equation}
	where $(\alph{OutEquations}\stepcounter{OutEquations})$ is due to the i.i.d nature of samples. The second term writes  as:
	\begin{equation}
	\begin{aligned}
	& I(I_B;Y^n|I_A) = \sum_{i=1}^n \left[I(I_BX^{i-1};Y_i|I_AY_{i+1}^n) - I(X^{i-1};Y_i|I_AI_BY_{i+1}^n)\right]\\
	&= \sum_{i=1}^n \left[I(X^{i-1};Y_i|I_AY_{i+1}^n) +I(I_B;Y_i|I_AX^{i-1}Y_{i+1}^n) - I(X^{i-1};Y_i|I_AI_BY_{i+1}^n)\right]\\
	&= \sum_{i=1}^n \left[I(I_B;Y_i|I_AX^{i-1}Y_{i+1}^n) + I(X_i;Y_{i+1}^n|I_AX^{i-1}) - I(X^{i-1};Y_i|I_AI_BY_{i+1}^n)\right] \ ,
	\end{aligned}
	\end{equation}
	where the final step is due to identity~\eqref{eq:Csiszar-Korner}. These inequalities lead to
	\begin{align}
	nR \geq  \sum_{i=1}^n \big[I(I_AX^{i-1}Y_{i+1^n};X_i) &+ I(I_B;Y_i|I_AX^{i-1}Y_{i+1}^n) \\
	&- I(X^{i-1};Y_i|I_AI_BY_{i+1}^n)\big] \ .\nonumber
	\end{align}
	Using the same definitions for the auxiliary RVs as above, this result can be expressed as follows:
	\begin{equation}
	R \geq I(\hat{U}_Q;X_Q|Q) + I(V_Q;Y_Q|\hat{U}_Q,Q) - T \ ,
	\end{equation}
	and thus, the following region is an outer bound:
	\begin{equation}\label{eq:tRegion}
	\begin{cases}
	\Delta_n \leq I(U;Y) + I(V;X|U) -T \ ,\\
	R \geq I(U;X) + I(V;Y|U) -T \ ,
	\end{cases}
	\end{equation}
	where $(U,V)$ are auxiliary RVs that respect the required Markov chains in~\eqref{Eq:TestingAgainstIndependence}. It is  left to show that~\eqref{eq:tRegion} is equivalent or stricter than:
		\begin{equation}\label{eq:FinalRegion}
		\begin{cases}
		\Delta_n \leq I(U;Y) + I(V;X|U) \ ,\\
		R \geq I(U;X) + I(V;Y|U) \ .
		\end{cases}
		\end{equation}
 That is, all pairs $(R,\Delta_n)$ that are forbidden in the region in \eqref{eq:tRegion} are also forbidden in~\eqref{eq:FinalRegion}. In order to do so we use \emph{Fourier-Motzkin} elimination~\cite{Schrijver-1998} over $T\geq 0$. 
		By removing $T$,  we get:
		\begin{equation}
		\begin{cases}
		\Delta_n \leq  I(U;Y) + I(V;X|U) \ , \\
		R \geq I(U;X) + I(V;Y|U) - I(U;Y) - I(V;X|U) + \Delta_n  \ ,
		\end{cases}
		\end{equation}
		and using the Markovian relations between the different RVs we obtain:
		\begin{equation}\label{eq:AlmostFinalRegion}
		\begin{cases}
		\Delta_n \leq  I(U;Y) + I(V;X|U) \ , \\
		R \geq I(U;X|Y) + I(V;Y|UX) + \Delta_n  \ .
		\end{cases}
		\end{equation}
		 In order to show the equivalence between the two regions, we need to check the extremal points. The point where $\Delta_n = 0$ is trivial, as $R = 0$ is optimal under both regions. When checking $\Delta_n =  I(U;Y) + I(V;X|U)$ we have:
		\begin{align}
		R &\geq I(U;X|Y) + I(V;Y|UX) + I(U;Y) + I(V;X|U)\\
		 &= I(U;X) + I(V;Y|U) \ ,\nonumber
		\end{align}
		which completes the proof of the weak unfeasibility.
\end{proof}

\begin{remark}
We conjecture that in contrast to the unidirectional testing problem~\cite{Ahlswede-Csiszar-1986}, the strong unfeasibility property --implying that the error exponent does not depend on $\epsilon$-- does not hold for the collaborative hypothesis testing problem. A possible reason for this failure is that such a property heavily relies on the Blowing Up lemma (see Lemma~\ref{Lemma:BlowingUp}) which does not hold conditioned on arbitrary probability events (e.g. the corresponding event induced from the first information layer).
\end{remark}

\section{Collaborative Hypothesis Testing with Zero Rate}\label{Sec:ZeroRate}
We now consider another special case of Proposition~\ref{Prop:ManyRounds}, whereby testing is done over two general hypotheses, but the total exchange rate is zero. It is worth mention that zero-rate does not mean that \emph{no information exchange} is possible, but rather that the size of the codebook grows slower than exponentially with the blocklength $n$, as stated in the following proposition. 
 \begin{theorem}[Necessary and sufficient conditions under zero-rate] \label{Prop:TestingWithZeroRate}
 	Let $P_{XY}$ and $P_{\bX\bY}$ be any probability measures such that 
	$\textrm{supp}(P_{\bX\bY})=\textrm{supp}(P_{XY})$ $= \mathcal{X}\times \mathcal{Y}$. Assume the total exchange rate $R=0$, that is:
 	\begin{equation}\label{eq:ConstraintZeroRate}
 	\sum \limits_{k=1}^K  \log |  f_{[k]} | +  \sum \limits_{k=1}^K  \log   |  g_{[k]} |  \equiv o(n) \ ,
 	\end{equation}
 	the optimal error exponent to the probability of Type II is given by
 	\begin{align}\label{eq:ExponentZeroRate}
 	\lim\limits_{n \to \infty} -\frac{1}{n} \log \beta_n (R = 0,\epsilon\, | K) &= \\
 \min\limits_{\tX\tY \in \mathscr{L}_0(X,Y)}\mathcal{D}(P_{\tX\tY}\|P_{\bX\bY}) &\coloneqq E(R=0) \ ,\ \, \forall\,0<\epsilon<1\ ,\nonumber
 	\end{align}
 	where $\mathscr{L}_0(X,Y) \coloneqq \big\{\tX\tY: P_{\tX}=P_{X}\, , \,P_{\tY} = P_Y\big\}$.
 \end{theorem}
It is worth mentioning that the same expression \eqref{eq:ExponentZeroRate} was proven in~\cite{Han-1987} to be feasible based on \emph{unidirectional one bit exchange}, i.e., $ |  f_{[1]} | = 2$,  $|  g_{[1]} |  = 0$. This observation implies that when zero-rate is enforced, not only data exchanges do not help, but only one bit of exchange is enough. In addition, note that this is a \emph{strong unfeasability} result, as the optimal exponent for $\beta_n$ is not dependent on the constraint $\epsilon$ over the error probability of Type I.
 
 \begin{proof}[Proof of Theorem~\ref{Prop:TestingWithZeroRate}]
From the expression of the error exponent in~\eqref{eq:ExponentZeroRate}, it is clear that it is enough to show the result for $K=1$, since it is  feasible with one round and the extension of the unfeasibility proof  is straightforward.
 	We start by proving the feasibility of the error exponent in \eqref{eq:ExponentZeroRate}  and then, we prove the unfeasibility result using methods similar to the ones in~\cite{Shalaby-Papamarcou-1992} for the case of a unidirectional exchanges.
 	
 	\subsubsection*{Feasibility}
 As the error exponent in \eqref{eq:ExponentZeroRate} is feasible  with single-side exchange, we use Proposition~\ref{Prop:Main} setting $V = \phi$. Thus, a feasible error exponent for zero-rate, as defined in Theorem~\ref{Prop:TestingWithZeroRate}:
 	\begin{equation}
\liminf\limits_{n \to \infty} -\frac{1}{n} \log \beta_n (R = 0,\epsilon\, | K)\geq \max_{\mathscr{S}(R=0)}\min_{\mathscr{L}(U,X,Y)} \cD(P_{\tU\tX\tY}||P_{\bU\bX\bY}) \ ,
 	\end{equation}
 	where $\mathscr{S}$ and $\mathscr{L}$ are the sets defined in Proposition~\ref{Prop:Main}. Using the chain rule for KL divergence, this exponent can be bounded as follows:
 	\begin{equation}
 	\begin{aligned}\label{eq-missing-inequality}
\max_{\mathscr{S}(R=0)}&\min_{\mathscr{L}(U,X,Y)} \cD(P_{\tU\tX\tY}||P_{\bU\bX\bY})\\
 	&= \max_{\mathscr{S}(R=0)}\min_{\mathscr{L}(U,X,Y)} \left[\cD(P_{\tX\tY}||P_{\bX\bY}) + \cD(P_{\tU|\tX\tY}||P_{\bU|\bX\bY}|P_{\tX\tY})\right]\\
 	&= \max_{\mathscr{S}(R=0)}\min_{\mathscr{L}_0(X,Y)}\left[\cD(P_{\tX\tY}||P_{\bX\bY}) + \min\limits_{P_{\tU|\tX\tY}}\cD(P_{\tU|\tX\tY}||P_{\bU|\bX\bY}|P_{\tX\tY})\right] \\
 	&\geq \min_{\mathscr{L}_0(X,Y)}\cD(P_{\tX\tY}||P_{\bX\bY}) \ .
 	\end{aligned}
 	\end{equation}
 	Here, the minimum over $P_{\tU|\tX\tY}$ is such that $\tU\tX\tY \in \mathscr{L}(U,X,Y)$, and the final inequality is due to the non-negativity of the KL divergence. 

 	\subsubsection*{Strong unfeasibility}
We now prove the optimality of Theorem~\ref{Prop:TestingWithZeroRate}, by showing that the error exponent of $\beta_n(R=0,\epsilon)$ does not depend on $\epsilon \in (0,1)$, and that \eqref{eq:ExponentZeroRate} cannot be beaten. We follow a similar approach to~\cite{Shalaby-Papamarcou-1992}, which addressed this proof for the case of unidirectional exchanges.
 	
 Let $f_{[1]}: \mathcal{X}^n  \rightarrow \{1,\dots,|  f_{[1]} |\}$  and $g_{[1]} : \mathcal{Y}^n \times \{1,\dots,| f_{[1]}| \}\rightarrow\{1,\dots, |  g_{[1]} |\}$ be  the encoding functions at node $A$ and $B$, respectively, and let $\phi\big(X^n, g_{[1]} (Y^n, $ $f_{[1]} (X^n))\big) \in \{0,1\}$ be the decoding function at node $A$. Define sets:
 	\begin{equation*}
 	\begin{aligned}
	\mathcal{C}_{ij} &\coloneqq \big \{\vct{x} \in \cX^n: f_{[1]}(\vct{x}) = i\, \textrm{ and } \phi(\vct{x},j)= 0\big\}\ ,  \ 
	 \cC_i \coloneqq \bigcup\limits_{i=1}^{|  f_{[1]} |}  \cC_{ij} \ ,\\
	\mathcal{F}_{ij} &\coloneqq \big\{\vct{y} \in \mathcal{Y}^n: g_{[1]}(\vct{y},i) = j\big\} \ , \  (i,j) \in \{1,\ldots,|  f_{[1]} |\}\times \{1,\ldots,|  g_{[1]} |\} \ .
	\end{aligned}
 	\end{equation*}
 	Note that $\cC_{ij}$ (respectively, $\cF_{ij}$) cannot be said to be pairwise disjoint in $\cX^n$ (respectively, $\cY^n$) while the sets $\cC_i$ are pairwise disjoint. Similarly, for each index $i_0$, the sets $\cF_{i_0j}$ are disjoint. The acceptance set of $H_0$ can be expressed by
 	\begin{equation}
 	\cA_n \coloneqq  \bigcup\limits_{i=1}^{|  f_{[1]} |}\bigcup\limits_{j=1}^{|  g_{[1]} |} \cC_{ij}\times\cF_{ij} \ .
 	\end{equation}
 	That is, if $(\vct{x},\vct{y}) \in \cA_n$, $\phi\big(\vct{x}, g_{[1]} (\vct{y}, f_{[1]} (\vct{x}))\big) = 0$ and otherwise, the result is $H_1$. By the definition, $P_{XY}^n(\cA_n^c) \leq \epsilon$, or equivalently 
 	\begin{equation}\label{eq:LimitSetA}
 	P_{XY}^n(\cA_n) = P_{XY}^n\left(\bigcup\limits_{i=1}^{|  f_{[1]} |}\bigcup\limits_{j=1}^{|  g_{[1]} |} \cC_{ij}\times\cF_{ij}\right)> 1 - \epsilon \ .
 	\end{equation}
	Since the sets $\mathcal{B}_i \coloneqq \bigcup\limits_{j=1}^{|  g_{[1]} |} \cC_{ij}\times\cF_{ij}$ are disjoint,  by relying on \eqref{eq:LimitSetA} and on the size $|  f_{[1]} |$, there exists an index $i_0$ such that 
 	\begin{equation}\label{eq:LimitSetB}
 	P_{XY}^n\left(\bigcup\limits_{j=1}^{|  g_{[1]} |} \cC_{i_0j}\times\cF_{i_0j}\right)\geq \frac{1 - \epsilon}{|  f_{[1]} |} \ .
 	\end{equation}
 	As the sets $F_{i_0j}$ are disjoint, there exists an index $j_0$ such that
 	\begin{equation}
 	P_{XY}^n(\cC_{i_0j_0}\times\cF_{i_0j_0})\geq \frac{1 - \epsilon}{|  f_{[1]} ||  g_{[1]} |}  \ .
 	\end{equation}
 	Letting $\cC \equiv \cC_{i_0j_0}$ and $\cF \equiv \cF_{i_0j_0}$, we rewrite this as:
 	\begin{equation}
 	P_{XY}^n(\cC\times\cF) \geq \frac{1-\epsilon}{|  f_{[1]} ||  g_{[1]} |} \equiv  \exp(-n\delta_n) \ ,
 	\end{equation}
 	with $\delta_n \equiv  \frac{1}{n}\log\left(|  f_{[1]} ||  g_{[1]} |\right) - \frac{1}{n}\log(1-\epsilon)$.
 	As the log-function is monotonic and both $|  f_{[1]} |$ and $|  g_{[1]} |$ are non-negative, expression~\eqref{eq:ConstraintZeroRate} implies that $\log |  f_{[1]} |=o(n)$ and $\log |  g_{[1]} | =o(n)$  and thus $\delta_n =o(1)$.
 	
 Having shown that there exist sets $\cC$ and $\cF$, such that $\cC\times\cF \in \cA_n$, and the probability $P_{XY}(\cC\times\cF)$ does not approach $0$ exponentially with $n$, the rest of the proof follows along the lines in~\cite{Shalaby-Papamarcou-1992}. For the sake of completeness, this proof is completed in Appendix~\ref{Apen:ZeroRate}.
 \end{proof}

\appendix
\section{Technical Definitions and Lemmas} \label{Apen:typicality}

In this appendix, we revise fundamental notions and properties of \emph{method of types}~\cite{Csiszar-1998}, which are extensively used through this paper. 

\begin{definition}[Types~\cite{Csiszar-Korner-2011}]
	The type of a sequence $\vct{x} \in \cX^n$ is the measure $\hat{P}_X$ on $\mathcal{X}$ defined by
	$
	\hat{P}_X(a) \coloneqq  \frac{1}{n} N(a|\vct{x}) \ ,\quad \forall a \in \cX\ ,
$
	where $N(a|\vct{x})$ is the counting measure of the letter $a$ in $\vct{x}$.
	The \emph{joint type} of a pair $(\vct{x},\vct{y}) \in \cX^n\times\cY^n$ is the empirical  measure $\hat{P}_{XY}$ on $\cX\times\cY$ such that
	\begin{equation}
	\hat{P}_{XY}(a,b) \coloneqq \frac{1}{n} N(a,b|\vct{x},\vct{y}) \ ,\quad \forall (a,b) \in \cX\times\cY \ ,
	\end{equation}
	where $N(a,b|\vct{x},\vct{y})$ is  the joint counting measure of the pair $(a,b)$ in $(\vct{x},\vct{y})$.
\end{definition}
	
	\begin{definition}[Conditional Types~\cite{Csiszar-Korner-2011}]
		The vector $\vct{y} \in \cY^n$ is said to have \emph{conditional type} $V:\mathcal{X}\mapsto \mathcal{P}_n(\mathcal{Y})$ given $\vct{x} \in \cX^n$ if 
		\begin{equation}
		N(a,b|\vct{x},\vct{y}) = N(a|\vct{x})V(b|a)\ , \quad \forall (a,b) \in \cX\times\cY \ ,
		\end{equation}
		where $V$ is a stochastic mapping.
	\end{definition}
	
\begin{lemma}[Type Counting]\label{Lemma:TypeCounting}
Let  $\cP_n(\mathcal{X})$ be the set of all possible types of sequences in $\cX^n$. Then, $|\cP_n(\mathcal{X}) | \leq (n+1)^{|\cX|} \ .$
\begin{proof}
Refer to  reference~\cite[Lemma 2.2]{Csiszar-Korner-2011}.
\end{proof}

\end{lemma}
\begin{lemma}\label{Lemma:SizeTypicalSet}
	For any type $\hat{P}\in \cP_n(\mathcal{X})$ of sequences in $\cX^n$, denote by $\cT_{[\hat{P}]}$ the set of all sequences with this type. Then,
	\begin{equation}
	(n+1)^{-|\cX |} \exp \big[n H(\hat{P}) \big] \leq |\cT_{[\hat{P}]}| \leq \exp\big[nH(\hat{P})\big] \ .
	\end{equation}
	In a similar fashion, for every $\vct{x} \in \cX^n$ and stochastic mapping $V:\mathcal{X}\mapsto \mathcal{P}_n(\mathcal{Y})$, let $\cT_{[V]}(\vct{x})$ be the set of all sequences $\vct{y} \in \cY^n$ with the conditional type $V$ given $\vct{x}$. Then,
	\begin{equation}
	(n+1)^{-|\cX||\cY|}\exp\big[ nH(V|\hat{P})\big] \leq |\cT_{[V]}(\vct{x})|\leq \exp\big[nH(V|\hat{P})\big] \ ,
	\end{equation}
	where $H(V|\hat{P})$ is the conditional entropy function,
	\begin{equation}
	H(V|\hat{P}) = \sum\limits_{x \in \cX} \hat{P}(x)H(V(\cdot|x)) \ .
	\end{equation}
\end{lemma}
\begin{proof}
Refer to reference~\cite[Lemma 2.3, Lemma 2.5]{Csiszar-Korner-2011}.
\end{proof}

\begin{lemma}[Inaccuracy]\label{Lemma:ProbabilityByType}
	Let $\hat{P} \in \cP_n(\mathcal{X})$ be the type of $\vct{x} \in \cX^n$ ($X^{(n)} \sim \hat{P} $ is referred to as the \emph{type variable}). Then, for any RV $X$ on $(\cX,\mathcal{B}_{\cX},P_X)$, 
	\begin{equation}
	P_{X}^n(X^n=\vct{x}) = \exp \Bigl\{-n\left[H(\hat{P}) + \cD(\hat{P}\|P_X)\right]\Bigr\} \ .
	\end{equation} 
\end{lemma}
	\begin{proof}
Refer to reference~\cite[Lemma 3]{Han-1987},\cite[Lemma 2.6]{Csiszar-Korner-2011}.
\end{proof}

\begin{definition}[$\delta$-Typicality~\cite{Han-1987}]\label{Def:DeltaTypicality}
	Let  $\delta>0$, an $n$-sequence $\vct{x}$ is called $\delta$-typical, denoted by $\cT_{[X]_\delta}$, if
$|N(a|\vct{x}) - nP_X(a)| \leq \mathcal{O}(\delta), \quad \forall a \in \cX \ ,
$ and $\hat{P}_X \ll P_X$. Jointly $\delta$-typical $\cT_{[XY]_\delta}$ and conditionally $\delta$-typical sequences $\cT_{[Y|X]_\delta}(\vct{x})$ are defined in a similar manner.
\end{definition}
	
	\begin{lemma}\label{Lemma:JointConditionalTypicality}
		Let $\cT_{[X]_\delta}$, $\cT_{[XY]_\delta}$ and $\cT_{[Y|X]_\delta}$ denote the sets of typical, jointly typical and conditionally typical sequences, respectively. For any $\vct{x} \in \cT_{[X]_\delta}$ and $\vct{y} \in \cT_{[Y|X]_{\delta'}}$, then $(\vct{x},\vct{y}) \in \cT_{[XY]_{\delta + \delta'}}$. Moreover, $\vct{y} \in \cT_{[Y]_{\delta''}}$, with $\delta'' \coloneqq  (\delta+\delta')|\cX|$.
\end{lemma}
\begin{proof}
	Refer to reference~\cite{Csiszar-Korner-2011}.
\end{proof}

     \begin{lemma}[Generalized Markov Lemma]
     	Let $p_{UXY}\in\mathcal{P}\left(\mathcal{U}\times\mathcal{X}\times\mathcal{Y}\right)$ be a probability measure that satisfies: $U \mkv X \mkv Y$. Consider $(\vct{x},\vct{y})\in\mathcal{T}^n_{[XY]_{\epsilon'}}$ and random vectors ${U}^n$  generated according to:
      	\begin{equation}
     	\label{eq:u_dist}
     	\Pr\left\{{U}^n=\vct{u}\big|{U}^n\in\mathcal{T}_{[U|X]_{\epsilon''}}^n(\vct{x}), \vct{x},\vct{y}\right\}=\frac{\mathds{1}\left\{{u}^n\in\mathcal{T}_{[U|X]_{\epsilon''}}^n(\vct{x})\right\}}{\big|\mathcal{T}_{[U|X]_{\epsilon''}}^n(\vct{x})\big|}\ .
     	\end{equation}
     	For sufficiently small $\epsilon,\epsilon',\epsilon''>0$,       	
	\begin{equation}
     	\Pr\left\{{U}^n\notin\mathcal{T}^n_{[U|XY]_{\epsilon}}(\vct{x},\vct{y})\Big|{U}^n\in\mathcal{T}^n_{[U|X]_{\epsilon''}}(\vct{x}),\vct{x},\vct{y}\right\}\equiv \mathcal{O}\left(c^{-n}\right)
     	\end{equation}
     	 holds uniformly on $(\vct{x},\vct{y})\in\mathcal{T}^n_{[XY]_{\epsilon'}}$ where $c>1$.
     	\label{Lemma:markov}
     	\end{lemma}
     	
     	\begin{proof}
     		Refer to reference \cite{Piantanida-Vega-Hero-2014}.
     	\end{proof}
	
	\begin{lemma}\label{Lemma:SizeDeltaTypicalSet}
For every probability measure $P_X \in \cP(\mathcal{X})$ and stochastic mapping $W:\mathcal{X}\mapsto \mathcal{P}(\mathcal{Y})$, there exist sequences $(\varepsilon_n)_{n \in \mathbb{N}_+},(\varepsilon^\prime_n)_{n \in \mathbb{N}_+} \to 0$ as  $n \to \infty$ satisfying:
		\begin{equation}
		 \left|\frac{1}{n} \log |\cT_{[X]_\delta}| - H(X)\right| \leq \varepsilon_n \ , \  \left|\frac{1}{n} \log |\cT_{[Y|X]_\delta}(\vct{x})| - H(Y|X)\right| \leq \varepsilon_n \ , 
\end{equation}
		 for each $\vct{x} \in \cT_{[X]_\delta}$ where $\varepsilon_n\equiv \mathcal{O}(n^{-1} \log n) $, and 
		\begin{equation}
		P_X^n\big(\cT_{[X]_\delta}\big) \geq 1- \varepsilon^\prime_n \ ,\ W^n\big(\cT_{[Y|X]_\delta}(\vct{x})| X^n=\vct{x}\big) \geq 1 - \varepsilon^\prime_n \ ,
		\end{equation}
		for all $\vct{x} \in \mathcal{X}^n$ where $\varepsilon_n^\prime\equiv \mathcal{O}\left(\frac{1}{n\delta^2}\right) $, provided that   $n$ is sufficiently large.
\end{lemma}
	\begin{proof}
Refer to reference~\cite[Lemma 2.13]{Csiszar-Korner-2011}.
\end{proof}

\section{}\label{Apen:MarkovChains}
As a part of the weak unfeasibility  part of the proof of Theorem~\ref{Theo:AgainstIndependenceOptimality}, two Markov chains are necessary:
\begin{equation}
\begin{cases}
\hat{U}_i \mkv X_i \mkv Y_i \ , \ \forall\, i=[1:n]\\
V_i \mkv (\hat{U}_i,Y_i) \mkv X_i \ , \ \forall\, i=[1:n].
\end{cases}
\end{equation}
Using the chosen RVs from~\eqref{eq:ChosenRandomVariables}, these Markov chains are represented by
\begin{equation}\label{eq:DesiredMarkovChains}
\begin{cases}
(I_A, X^{i-1}, Y_{i+1}^n) \mkv X_i \mkv Y_i \ ,\ \forall\, i=[1:n]\\
I_B \mkv (I_A, X^{i-1}, Y_{i}^n) \mkv X_i  \ , \ \forall\, i=[1:n].
\end{cases}
\end{equation}
In order to check this, we use the next lemma.

\begin{lemma}\label{Lemma:Kaspi}
	Let $A_1,A_2,B_1,B_2$ be RVs with joint probability measure $P_{A_1A_2B_1B_2}= P_{A_1B_1}P_{A_2B_2}$ and assume  that $\{f^i\}_{i=1}^k,\{g^i\}_{i=1}^k$ are any collection of $P$-measurable mappings with domain structure given by: 
	\begin{align}
	& f^1(A_1,A_2); f^2(A_1,A_2,g^1); \ldots; f^k(A_1,A_2,g^1,\ldots,g^{k-1}) \ , \\
	& g^1(B_1,B_2,f^1); g^2(B_1,B_2,f^1,f^2); \ldots ; g^k(B_1,B_2,f^1,\ldots, f^k) \ .
	\end{align}
	Then,
	\begin{equation}
	I(A_2;B_1|f^1,f^2,\ldots,f^k,g^1,g^2,\ldots,g^k,A_1,B_2) = 0 \ .
	\end{equation}
\end{lemma}
\begin{proof}
Refer to reference~\cite[Lemma 1]{Kaspi-1985}.
\end{proof}

In order to prove the first Markov chain, we simply let: 
	\begin{equation}
	\begin{cases}
	A_1 \coloneqq  X_i, \quad B_1  \coloneqq Y_i \ ,\\
	A_2  \coloneqq (X^{i-1},X_{i+1}^n,Y_{i+1}^n)\ , \quad B_2  \coloneqq Y^{i-1} \ .
	\end{cases}
	\end{equation}
	It can be easily verified that $P_{A_1A_2B_1B_2}= P_{A_1B_1}P_{A_2B_2}$, which stems directly from the i.i.d. nature of the samples. Thus, according to Lemma~\ref{Lemma:Kaspi}:
	\begin{equation}
	\begin{aligned}
	0 &= I(X^{i-1}X_{i+1}^nY_{i+1}^n;Y_i|X_iY^{i-1})\\
	&= I(X^{i-1}X_{i+1}^nY^{i-1}Y_{i+1}^n;Y_i|X_i) - I(Y^{i-1};Y_i|X_i) \ ,
	\end{aligned}
	\end{equation}
which shows   the Markov chain:
\begin{equation}
(X^{i-1},X_{i+1}^n,Y^{i-1},Y_{i+1}^n) \mkv X_i \mkv Y_i \ , \ \forall\, i=[1:n].
\end{equation}
As $I_A \coloneqq f_{[1]}(X^n)$, the following Markov chain is also true:
\begin{equation}
(I_A,X^{i-1},Y_{i+1}^n) \mkv X_i \mkv Y_i  \ , \ \forall\, i=[1:n]
\end{equation}
which proves the first Markov chain in~\eqref{eq:DesiredMarkovChains}. As for the second one, we let:
\begin{equation}
\begin{cases}
A_1 \coloneqq X^{i-1}\ , \quad B_1 \coloneqq Y^{i-1} \ ,\\
A_2 \coloneqq (X_i,X_{i+1}^n)\ , \quad B_2 \coloneqq (Y_i,Y_{i+1}^n) \ .
\end{cases}
\end{equation}
Under this choice, $I_A \coloneqq  f_{[1]}(A_1,A_2)$ and thus,
\begin{equation}
I(X_iX_{i+1}^n;Y^{i-1}|I_AX^{i-1}Y_iY_{i+1}^n) = 0 \ , \ \forall\, i=[1:n].
\end{equation}
The later identity proves the following Markov chain:
\begin{equation}
(X_i,X_{i+1}^n) \mkv (I_A,X^{i-1},Y_i,Y_{i+1}^n) \mkv Y^{i-1} \ , \ \forall\, i=[1:n].
\end{equation}
As $I_B \coloneqq  g_{[1]}(I_A,Y^n)$, it also holds that:
\begin{equation}
X_i \mkv (I_A,X^{i-1},Y_{i}^n) \mkv I_B \ , \ \forall\, i=[1:n]
\end{equation}
which yields the desired Markov chain.



	\section{}\label{Apen:WeakConverse}
	
	\begin{proof}[Proof of Lemma~\ref{Prop:WeakConverse}]
For block-length $n$, given a code  characterized by the encoding mappings $f_{[1]},g_{[1]}$ at nodes $A$ and $B$ respectively, and a decoding mapping $\phi$ at node $A$. Let the acceptance region be denoted by 
\begin{equation}
\mathcal{A}_n \coloneqq \big \{(\vct{x},j) \in \cX^n\times\{1,\dots, |g_{[1]}|\}: g_{[1]}\big(\vct{y},f_{[1]}(\vct{x})\big)  = j, \ \vct{y} \in \cY^n ,\ \phi\big(\vct{x},j\big)= 0\big\}\ .  \ 
\end{equation}
Let $P$ and $Q$ denote the probabilities measures on $\cX^n\times\{1,\dots, |g_{[1]}|\}$ induced by  $H_0$ and $H_1$, respectively. From the \emph{log-sum inequality}~\cite{Csiszar-Korner-2011}, we have:
		\begin{align}
		&\mathcal{D}\left(P_{X^nI_AI_B}\|Q_{X^nI_AI_B}\right)=\mathcal{D}\left(P_{X^nI_B}\|Q_{X^nI_B}\right) \\
		&\geq (1-\alpha_n)\log\frac{1-\alpha_n}{\beta_n(R,\epsilon\,|K=1)} + \alpha_n\log\frac{\alpha_n}{1-\beta_n(R,\epsilon\,|K=1)} \ ,\nonumber
		\end{align}
		where $I_A \coloneqq  f_{[1]}(X^n)$, $I_B \coloneqq g_{[1]}(I_A,Y^n)$, $\alpha_n(R|K=1) \coloneqq  P(\mathcal{A}_n^c)\leq \epsilon$ and $\beta_n(R,\epsilon\,|K=1) \coloneqq  Q(\mathcal{A}_n)$. Through some algebra this yields:
		\begin{equation}
		\begin{aligned}
		\mathcal{D}\left(P_{X^nI_AI_B}\|Q_{X^nI_AI_B}\right) \geq (1-\alpha_n) \log\frac{1}{\beta_n(R,\epsilon\,|K=1)} -h_2(\alpha_n) \ ,
		\end{aligned}
		\end{equation}
		where $h_2(p)\coloneqq -p\log p - (1-p) \log (1-p)$ is the \emph{binary entropy} function. By assumption $\epsilon \to 0$ as $n \to \infty$, one conclude that for $n$ large enough
		\begin{equation}
		-\frac1n \log \beta_n(R,\epsilon\,|K=1)\leq  \frac1n \mathcal{D}\left(P_{X^nI_AI_B}\|Q_{X^nI_AI_B}\right) - \delta_n \ ,
		\end{equation}
		with $\delta_n \to 0$ as $n\to \infty$. Using the chain rule, we continue to get:
		\begin{align}
		\mathcal{D}\big(P_{X^nI_AI_B} \|Q_{X^nI_AI_B}\big)& \overset{(\alph{InEquations}\stepcounter{InEquations})}{=} I(I_B;X^n|I_A) + \mathcal{D}\left(P_{I_B|I_A}\|Q_{I_B|I_A}|P_{I_A}\right)\nonumber\\
		& \overset{(\alph{InEquations}\stepcounter{InEquations})}{\leq}I(I_B;X^n|I_A) + \mathcal{D}\left(P_{Y^nI_AI_B}\|Q_{Y^nI_AI_B}\right)\nonumber\\
		&\overset{(\alph{InEquations}\stepcounter{InEquations})}{=} I(I_B;X^n|I_A) + \mathcal{D}\left(P_{Y^nI_A}\|P_{Y}^n P_{I_A}\right)\nonumber\\
		&= I(I_B;X^n|I_A) + I(I_A;Y^n) \ .\nonumber
		\end{align}
		Here, $(\alph{OutEquations}\stepcounter{OutEquations})$ and $(\alph{OutEquations}\stepcounter{OutEquations})$ stem from the chain rule for the KL-divergence, and $(\alph{OutEquations}\stepcounter{OutEquations})$ is due to the fact that we consider the case of testing against independence. With this, the \emph{weak unfeasibility} proof  is completed. 
	\end{proof}

\section{}\label{Apen:ZeroRate}
\begin{proof}[Complementary proof of Theorem~\ref{Prop:TestingWithZeroRate}]
We now complete  the proof of the strong unfeasibility to Theorem~\ref{Prop:TestingWithZeroRate}. To this end, we recall that we showed there exist sets $\cC\subset \cX^n$ and $\cF\subset \cY^n$ such that $\cC\times\cF \in \cA_n$, and $P^n_{XY}(\cC\times\cF) \geq \exp(-n\delta_n)$, with $\delta_n \to 0$ as $n \to \infty$. We now evoke the ``Blowing-Up'' Lemma:
\begin{lemma}[Blowing-up Lemma]\label{Lemma:BlowingUp}
Let $Y^n=(Y_1,\dots,Y_n)$ be independent random variables in $(\cY^n,\mathcal{B}_{\cY^n})$ distributed according to $W^n(Y^n|X^n=\vct{x})$ for some fixed vector $\vct{x}\in \cX^n$ and a stochastic mapping $W:\mathcal{X}\mapsto \mathcal{P}(\mathcal{Y})$ and let  $\delta_n \to 0$ be a given sequence. There exist sequences  ${k_n}\equiv o(n)$ and $\gamma_n\equiv o(1)$,  such that for every subset $\mathcal{A}_n \subset \cY^n$: 
\begin{equation}
W^n(\mathcal{A}_n|X^n=\vct{x}) \geq \exp(-n\delta_n)\ \textrm{implies} \ W^n\big(\Gamma^{k_n}\mathcal{A}_n|X^n=\vct{x}\big) \geq 1 - \gamma_n
\end{equation}
where $\Gamma^{k_n}\mathcal{A}_n$ denotes the $\Gamma^{k_n}$-neighborhood of the set $\mathcal{A}_n$ defined by
	\begin{equation}
	\Gamma^{k_n}\mathcal{A}_n \coloneqq  \left\{\hat{\vct{y}} \in \cY^n: \, \min\limits_{\vct{y} \in \mathcal{A}_n}\rho_n(\hat{\vct{y}},\vct{y}) \leq k_n\right\} \ ,
	\end{equation}
	where $\rho_n(\hat{\vct{y}},\vct{y}) \coloneqq  \sum\limits_{i=1}^n \mathbb{1}\{\hat{y}_i\neq y_i\}$ and $\mathbb{1}\{\hat{y}\neq y\}=1 $ if $\hat{y}\neq y$ or $=0$ otherwise. 
\end{lemma}
\begin{proof}
Refer to references~\cite{Margulis74,Ahlswede-Gacs-Korner}.
\end{proof}
The rest of the proof follows closely the steps taken in \cite{Shalaby-Papamarcou-1992}.
As $P^n_{XY}(\cC\times\cF) \geq \exp(-n\delta_n)$, clearly $P_X^n(\cC) \geq \exp(-n\delta_n)$ and $P^n_Y(\cF) \geq \exp(-n\delta_n)$. Using the non-conditional version of Lemma~\ref{Lemma:BlowingUp}, there exist sequences $k_n=o(n)$ and $\gamma_n=o(1)$  s.t.:
\begin{equation}\label{eq:AfterBlowingUp}
P_X^n\big(\Gamma^{k_n}\cC\big) \geq 1 - \gamma_n\ , \ \quad P_Y^n\big(\Gamma^{k_n}\cF\big) \geq 1 - \gamma_n \ ,
\end{equation}
where $k_n,\gamma_n$ only depend on $|\cX|,|\cY|$ and $\delta_n$, but not on $P_{XY}$. Equation~\eqref{eq:AfterBlowingUp} holds true if we change $P_X$ to $P_{\tX}$ and $P_Y$ to $P_{\tY}$, for some $\tX\tY \in \mathscr{L}_0$. As we wish to analyze the error probability  for fixed $n$,  during most of this proof we take the liberty to dismiss the subscript $n$ from $k_n$, for the sake of readability.

Using the fact $\Pr(A \cap B) \geq \Pr(A) + \Pr(B) - 1$ and~\eqref{eq:AfterBlowingUp}, we obtain:
\begin{equation}\label{eq:Combine1}
P_{\tX\tY}^n\big(\Gamma^k\cC\times\Gamma^k\cF\big) \geq P_{\tX}^n\big(\Gamma^k\cC\big) + P_{\tY}^n\big(\Gamma^k\cF\big) - 1 \geq 1 - 2\gamma_n \ .
\end{equation}
Consider  the set of  $\eta$-typical sequences defined by $P_{\tX\tY}$. By Lemma~\ref{Lemma:SizeDeltaTypicalSet},
\begin{equation}\label{eq:Combine2}
P_{\tX\tY}^n(\cT_{[\tX\tY]\eta}) \geq 1 - \mathcal{O}\left(\frac{1}{n\eta^2}\right) = 1 - \mathcal{O}\left(n^{-\frac{1}{3}}\right) \ ,
\end{equation}
where the last equality is a result of the choice $\eta \equiv  \eta_n \coloneqq n^{-\frac{1}{3}}$. Combining~\eqref{eq:Combine1} and \eqref{eq:Combine2}, it is clear that for sufficiently large $n$,
\begin{equation}\label{eq:CapBiggerThanHalf}
P_{\tX\tY}^n\big(\Gamma^k\cC\times\Gamma^k\cF) \cap \cT_{[\tX\tY]_\eta}\big) \geq \frac{1}{2} \ .
\end{equation}
By the definition of the $\eta$-typical set (see Definition~\ref{Def:DeltaTypicality}), we have:
\begin{equation}
\cT_{[\tX\tY]_\eta} = \!\!\!\!\!\!\!\bigcup\limits_{\substack{P_{\hat{X}\hat{Y}} \in \mathcal{P}_n(\cX\times\cY)\\ |P_{\hat{X}\hat{Y}}-P_{\tX\tY}| \leq \eta\ , \ P_{\hat{X}\hat{Y}} \ll P_{\tX\tY}}} \!\!\!\!\!\!\! \cT_{[\hat{X}\hat{Y}]} \ ,
\end{equation} 
where $|P_{\hat{X}\hat{Y}}-P_{\tX\tY}| \leq \eta$ refers to the maximum over all the arguments in $\cX\times\cY$. As all elements of $\cT_{[\hat{X}\hat{Y}]}$ are equiprobable under an i.i.d measure, \eqref{eq:CapBiggerThanHalf} can be rewritten as
\begin{equation}\label{eq:SumBiggerThanHalf}
\sum\limits_{\substack{P_{\hat{X}\hat{Y}} \in \mathcal{P}_n(\cX\times\cY)\\ |P_{\hat{X}\hat{Y}}-P_{\tX\tY}| \leq \eta\ , \ P_{\hat{X}\hat{Y}} \ll P_{\tX\tY}}} P_{\tX\tY}^n\big(\cT_{[\hat{X}\hat{Y}]}\big) \frac{|(\Gamma^k\cC\times\Gamma^k\cF) \cap \cT_{[\hat{X}\hat{Y}]_\eta}|} {|\cT_{[\hat{X}\hat{Y}]_\eta}|} \geq \frac{1}{2} \ .
\end{equation}
As $P_{\tX\tY}^n(\cT_{[\hat{X}\hat{Y}]}) \leq 1$, by using the bound over the size of the set $\cP_n(\cX\times\cY)$ in Lemma~\ref{Lemma:TypeCounting}, there must be \emph{at least one type} $\cT_{[\hat{X}\hat{Y}]}$, for which
\begin{equation}\label{eq:FracBiggerThanHalf}
\frac{|(\Gamma^k\cC\times\Gamma^k\cF) \cap \cT_{[\hat{X}\hat{Y}]_\eta}|} {|\cT_{[\hat{X}\hat{Y}]_\eta}|} \geq \frac{1}{2}(n+1)^{-|\cX|\cY|} = \frac{1}{2}\exp(-n\epsilon_n) \ ,
\end{equation}
with $\epsilon_n = \mathcal{O}(n^{-1}\log(n+1)) \to 0$ as $n \to \infty$.
The equiprobability property is also true for the probability measure implied by $H_1$, that is $P_{\bX\bY}$. Thus, 
\begin{equation}\label{eq:SameExponentH1}
\begin{aligned}
P_{\bX\bY}^n\big(\Gamma^k\cC\times\Gamma^k\cF\big) &\geq P_{\bX\bY}^n\big(\Gamma^k\cC\times\Gamma^k\cF\big) \cap \cT_{\hat{X}\hat{Y}})\\ &= P^n_{\bX\bY}(\cT_{\hat{X}\hat{Y}})\frac{|(\Gamma^k\cC\times\Gamma^k\cF) \cap \cT_{\hat{X}\hat{Y}}|}{|\cT_{\hat{X}\hat{Y}}|}\\ &\geq \frac{1}{2}\exp(-n\epsilon_n)P^n_{\bX\bY}(\cT_{\hat{X}\hat{Y}}) \ ,
\end{aligned}
\end{equation}
where the final inequality stems from \eqref{eq:FracBiggerThanHalf}.

Consider now an arbitrary element $(\vct{u},\vct{v}) \in \Gamma^k\cC\times\Gamma^k\cF$. By definition, there exist an element $(\vct{x},\vct{y}) \in \cC\times\cF$, such that $(u_i,v_i) \neq (x_i,y_i)$ at most in $2k$ locations. Thus,
\begin{equation}
P_{\bX\bY}^n(\vct{u},\vct{v}) = \prod\limits_{i=1}^n P_{\bX\bY}(u_i,v_i) \leq \rho^{-2k} \prod\limits_{i=1}^nP_{\bX\bY}(x_i,y_i) = \rho^{-2k}P^n_{\bX\bY}(\vct{x},\vct{y}) \ ,
\end{equation}
with $\rho = \min\limits_{(x,y) \in \cX\times\cY} P_{\bX\bY}(x,y)$, and we assume that $\rho > 0$ (which complies with the preliminaries of Theorem~\ref{Prop:TestingWithZeroRate}). As $(\vct{u},\vct{v})$ range over $\Gamma^k\cC\times\Gamma^k\cF$, each element $(\vct{x},\vct{y}) \in \cC\times\cF$ will be chosen as the closest neighbor at most $|\Gamma^k(\vct{x})|\times|\Gamma^k(\vct{y})|$ times. Thus,
\begin{equation}
P_{\bX\bY}^n\big(\Gamma^k\cC\times\Gamma^k\cF\big) \leq \rho^{-2k}|\Gamma^k(\vct{x})|\times|\Gamma^k(\vct{y})|P^n_{\bX\bY}(\cC\times\cF) \ .
\end{equation}
From \cite[Lemma 5.1]{Csiszar-Korner-2011} we have:
\begin{equation}
|\Gamma^k_n(\vct{x})| \leq \exp\left[n\left(h_2\left(\frac{k_n}{n}\right) + \frac{k_n}{n}\log|\cX|\right)\right] \equiv \exp(n\zeta'_n) \ ,
\end{equation}
with $h_2(\cdot)$ being the \emph{binary entropy} function and $\zeta'_n \to 0$ as $n \to \infty$. This implies that
\begin{equation}
P_{\bX\bY}^n\big(\Gamma^k\cC\times\Gamma^k\cF\big) \leq \exp(n\zeta_n)P^n_{\bX\bY}(\cC\times\cF) \ ,
\end{equation}
with $\zeta_n \coloneqq  2h_2\left(\frac{k_n}{n}\right) + \frac{k_n}{n}\log(|\cX|\cY|) - \frac{2k_n}{n}\log\rho \to 0$ as $n \to \infty$. Combining this with~\eqref{eq:SameExponentH1}, we finally get
\begin{align}
P_{\bX\bY}^n(\cC\times\cF) &\geq \exp(-n\zeta_n)P^n_{\bX\bY}\big(\Gamma^k\cC\times\Gamma^k\cF\big) \\
& \geq \frac{1}{2}\exp\left[-n(\zeta_n+\epsilon_n)\right]P^n_{\bX\bY}(\cT_{\hat{X}\hat{Y}})\nonumber\\
&\geq \frac{(n+1)^{|\cX|\cY|}}{2}\exp\left[-n\big(\mathcal{D}(P_{\hat{X}\hat{Y}}\|P_{\bX\bY}) + \zeta_n + \epsilon_n\big)\right]\nonumber\\
&\geq \exp\left[-n\big(\mathcal{D}(P_{\hat{X}\hat{Y}}\|P_{\bX\bY}) + \mu_n\big)\right] \ ,\nonumber
\end{align}
and $\mu_n \equiv \mu_n(\rho,\epsilon,M_n,N_n,|\cX|,|\cY|) \to 0$ as $n \to \infty$.

The previous conclusion is true for \emph{some type} $P_{\hat{X}\hat{Y}}$ over the range of all types that are $\eta$-typical for the measure $P_{\tX\tY}$. As the divergence functional $\mathcal{D}(\cdot\|\cdot)$ is convex and bounded, it is also uniformly continuous. It follows that we can find a sequence $\mu'_n \equiv  \mu'_n(\rho,|\cX|,|\cY|)$ such that $|P_{\hat{X}\hat{Y}} - P_{\tX\tY}| \leq \eta = o(n^{-\frac{1}{3}})$ implies that $|\mathcal{D}(P_{\hat{X}\hat{Y}}\|P_{\bX\bY}) - \mathcal{D}(P_{\tX\tY}\|P_{\bX\bY})| \leq \mu'_n$. Hence
\begin{equation}
P_{\bX\bY}^n(\cC\times\cF) \geq \exp\left[-n\big(\mathcal{D}(P_{\tX\tY}\|P_{\bX\bY}) + \mu_n + \mu'_n\big)\right] \ ,
\end{equation}
and consequently
\begin{align}
-\liminf\limits_{n \to \infty}\frac{1}{n}\log P^n_{\bX\bY}(\cA_n)& = -\lim\limits_{n \to \infty}\frac{1}{n}\log \beta_n(R=0,\epsilon\,|K=1) \\
& \leq \mathcal{D}(P_{\tX\tY}\|P_{\bX\bY}) \ ,\nonumber
\end{align}
and the RVs $\tX\tY$ are chosen from the set $\mathscr{L}_0$, which concludes the proof.
\end{proof}


%

\bibliographystyle{imsart-number}
\bibliography{HypothesisTestingBib}

\begin{thebibliography}{28}

\bibitem{Ahlswede-Burnashev-90}
\begin{barticle}[author]
\bauthor{\bsnm{Ahlswede},~\bfnm{Rudolf}\binits{R.}} \AND
  \bauthor{\bsnm{Burnashev},~\bfnm{MV}\binits{M.}}
(\byear{1990}).
\btitle{On minimax estimation in the presence of side information about remote
  data}.
\bjournal{The Annals of Statistics}
\bvolume{18}
\bpages{141--171}.
\end{barticle}
\endbibitem

\bibitem{Ahlswede-Csiszar-1986}
\begin{barticle}[author]
\bauthor{\bsnm{Ahlswede},~\bfnm{R.}\binits{R.}} \AND
  \bauthor{\bsnm{Csiszar},~\bfnm{I.}\binits{I.}}
(\byear{1986}).
\btitle{Hypothesis testing with communication constraints}.
\bjournal{Information Theory, IEEE Transactions on}
\bvolume{32}
\bpages{533-542}.
\bdoi{10.1109/TIT.1986.1057194}
\end{barticle}
\endbibitem

\bibitem{Ahlswede-Gacs-Korner}
\begin{barticle}[author]
\bauthor{\bsnm{Ahlswede},~\bfnm{Rudolf}\binits{R.}},
  \bauthor{\bsnm{G{\'a}cs},~\bfnm{Peter}\binits{P.}} \AND
  \bauthor{\bsnm{K{\"o}rner},~\bfnm{J{\'a}nos}\binits{J.}}
(\byear{1976}).
\btitle{Bounds on conditional probabilities with applications in multi-user
  communication}.
\bjournal{Zeitschrift f{\"u}r Wahrscheinlichkeitstheorie und verwandte Gebiete}
\bvolume{34}
\bpages{157--177}.
\end{barticle}
\endbibitem

\bibitem{Bucklew-Ney-91}
\begin{barticle}[author]
\bauthor{\bsnm{Bucklew},~\bfnm{JA}\binits{J.}} \AND
  \bauthor{\bsnm{Ney},~\bfnm{PE}\binits{P.}}
(\byear{1991}).
\btitle{Asymptotically optimal hypothesis testing with memory constraints}.
\bjournal{The Annals of Statistics}
\bvolume{18}
\bpages{982--998}.
\end{barticle}
\endbibitem

\bibitem{Chiyonobu-2001}
\begin{barticle}[author]
\bauthor{\bsnm{Chiyonobu},~\bfnm{Taizo}\binits{T.}}
(\byear{2001}).
\btitle{Hypothesis testing for signal detection problem and large deviations}.
\bjournal{Nagoya Mathematical Journal}
\bvolume{162}
\bpages{187--203}.
\end{barticle}
\endbibitem

\bibitem{Cover-69}
\begin{barticle}[author]
\bauthor{\bsnm{Cover},~\bfnm{Thomas~M}\binits{T.~M.}}
(\byear{1969}).
\btitle{Hypothesis testing with finite statistics}.
\bjournal{The Annals of Mathematical Statistics}
\bvolume{40}
\bpages{828--835}.
\end{barticle}
\endbibitem

\bibitem{Cover-Thomas-1991}
\begin{bbook}[author]
\bauthor{\bsnm{Cover},~\bfnm{Thomas~M}\binits{T.~M.}} \AND
  \bauthor{\bsnm{Thomas},~\bfnm{Joy~A}\binits{J.~A.}}
(\byear{1991}).
\btitle{Elements of information theory}.
\bpublisher{John Wiley \& Sons}, \baddress{New York}.
\end{bbook}
\endbibitem

\bibitem{Csiszar-1998}
\begin{barticle}[author]
\bauthor{\bsnm{Csisz\'ar},~\bfnm{I.}\binits{I.}}
(\byear{1998}).
\btitle{The Method of Types}.
\bjournal{Information Theory, IEEE Transactions on}
\bvolume{44}
\bpages{2505-2523}.
\end{barticle}
\endbibitem

\bibitem{Csiszar-Korner-2011}
\begin{bbook}[author]
\bauthor{\bsnm{Csiszar},~\bfnm{Imre}\binits{I.}} \AND
  \bauthor{\bsnm{K{\"o}rner},~\bfnm{J{\'a}nos}\binits{J.}}
(\byear{2011}).
\btitle{Information theory: coding theorems for discrete memoryless systems}.
\bpublisher{Cambridge University Press}.
\end{bbook}
\endbibitem

\bibitem{ElGamal-Kim-2011}
\begin{bbook}[author]
\bauthor{\bsnm{El~Gamal},~\bfnm{Abbas}\binits{A.}} \AND
  \bauthor{\bsnm{Kim},~\bfnm{Young-Han}\binits{Y.-H.}}
(\byear{2011}).
\btitle{Network information theory}.
\bpublisher{Cambridge University Press}.
\end{bbook}
\endbibitem

\bibitem{Han-1987}
\begin{barticle}[author]
\bauthor{\bsnm{Han},~\bfnm{Te}\binits{T.}}
(\byear{1987}).
\btitle{Hypothesis testing with multiterminal data compression}.
\bjournal{Information Theory, IEEE Transactions on}
\bvolume{33}
\bpages{759-772}.
\bdoi{10.1109/TIT.1987.1057383}
\end{barticle}
\endbibitem

\bibitem{Hellman-Cover-70}
\begin{barticle}[author]
\bauthor{\bsnm{Hellman},~\bfnm{Martin~E}\binits{M.~E.}} \AND
  \bauthor{\bsnm{Cover},~\bfnm{Thomas~M}\binits{T.~M.}}
(\byear{1970}).
\btitle{Learning with finite memory}.
\bjournal{The Annals of Mathematical Statistics}
\bvolume{41}
\bpages{765--782}.
\end{barticle}
\endbibitem

\bibitem{Kaspi-1985}
\begin{barticle}[author]
\bauthor{\bsnm{Kaspi},~\bfnm{A.}\binits{A.}}
(\byear{1985}).
\btitle{Two-way source coding with a fidelity criterion}.
\bjournal{Information Theory, IEEE Transactions on}
\bvolume{31}
\bpages{735-740}.
\bdoi{10.1109/TIT.1985.1057118}
\end{barticle}
\endbibitem

\bibitem{IT-2016}
\begin{barticle}[author]
\bauthor{\bsnm{{Katz}},~\bfnm{G.}\binits{G.}},
  \bauthor{\bsnm{{Piantanida}},~\bfnm{P.}\binits{P.}} \AND
  \bauthor{\bsnm{{Debbah}},~\bfnm{M.}\binits{M.}}
(\byear{2016}).
\btitle{{Distributed Binary Detection with Lossy Data Compression}}.
\bjournal{ArXiv e-prints}.
\bnote{{Submitted to Information Theory, IEEE Trans. on}}.
\end{barticle}
\endbibitem

\bibitem{Lehmann-2005}
\begin{bbook}[author]
\bauthor{\bsnm{Lehmann},~\bfnm{E.~L.}\binits{E.~L.}} \AND
  \bauthor{\bsnm{Romano},~\bfnm{J.~P.}\binits{J.~P.}}
\btitle{Testing Statistical Hypotheses}.
\bseries{Springer Texts in Statistics}.
\end{bbook}
\endbibitem

\bibitem{Margulis74}
\begin{barticle}[author]
\bauthor{\bsnm{Margulis},~\bfnm{G.~A.}\binits{G.~A.}}
(\byear{1974}).
\btitle{Probabilistic characteristics of graphs with large connectivity}.
\bjournal{Problemy Pereda\v ci Informacii}
\bvolume{10}
\bpages{101--108}.
\end{barticle}
\endbibitem

\bibitem{Naghshvar-Javidi-2013}
\begin{barticle}[author]
\bauthor{\bsnm{Naghshvar},~\bfnm{Mohammad}\binits{M.}} \AND
  \bauthor{\bsnm{Javidi},~\bfnm{Tara}\binits{T.}}
(\byear{2013}).
\btitle{Active sequential hypothesis testing}.
\bjournal{The Annals of Statistics}
\bvolume{41}
\bpages{2703--2738}.
\end{barticle}
\endbibitem

\bibitem{Nussbaum-Szkola-2009}
\begin{barticle}[author]
\bauthor{\bsnm{Nussbaum},~\bfnm{Michael}\binits{M.}} \AND
  \bauthor{\bsnm{Szko{\l}a},~\bfnm{Arleta}\binits{A.}}
(\byear{2009}).
\btitle{The Chernoff lower bound for symmetric quantum hypothesis testing}.
\bjournal{The Annals of Statistics}
\bvolume{37}
\bpages{1040--1057}.
\end{barticle}
\endbibitem

\bibitem{Piantanida-Vega-Hero-2014}
\begin{binproceedings}[author]
\bauthor{\bsnm{Piantanida},~\bfnm{P.}\binits{P.}},
  \bauthor{\bsnm{Rey~Vega},~\bfnm{L.}\binits{L.}} \AND
  \bauthor{\bsnm{Hero},~\bfnm{A.}\binits{A.}}
(\byear{2014}).
\btitle{A Proof of the Generalized Markov Lemma with Countable Infinite
  Sources}.
In \bbooktitle{Information Theory Proceedings (ISIT), 2014 IEEE International
  Symposium on}.
\end{binproceedings}
\endbibitem

\bibitem{Schrijver-1998}
\begin{bbook}[author]
\bauthor{\bsnm{Schrijver},~\bfnm{Alexander}\binits{A.}}
(\byear{1998}).
\btitle{Theory of linear and integer programming}.
\bpublisher{John Wiley \& Sons}.
\end{bbook}
\endbibitem

\bibitem{Shalaby-Papamarcou-1992}
\begin{barticle}[author]
\bauthor{\bsnm{Shalaby},~\bfnm{H.~M.~H.}\binits{H.~M.~H.}} \AND
  \bauthor{\bsnm{Papamarcou},~\bfnm{A.}\binits{A.}}
(\byear{1992}).
\btitle{Multiterminal detection with zero-rate data compression}.
\bjournal{Information Theory, IEEE Transactions on}
\bvolume{38}
\bpages{254-267}.
\bdoi{10.1109/18.119685}
\end{barticle}
\endbibitem

\bibitem{Shimokawa-han-amari-1994}
\begin{binproceedings}[author]
\bauthor{\bsnm{Shimokawa},~\bfnm{H.}\binits{H.}},
  \bauthor{\bsnm{Han},~\bfnm{T.}\binits{T.}} \AND
  \bauthor{\bsnm{Amari},~\bfnm{S.~I.}\binits{S.~I.}}
(\byear{1994}).
\btitle{Error Bound of Hypothesis Testing with Data Compression}.
In \bbooktitle{Inf. Theory, 1994 IEEE International Symposium on (ISIT)}
\bpages{114}.
\bdoi{10.1109/ISIT.1994.394874}
\end{binproceedings}
\endbibitem

\bibitem{Vega-Piantanida-Hero-2015}
\begin{barticle}[author]
\bauthor{\bsnm{Vega},~\bfnm{L.~R.}\binits{L.~R.}},
  \bauthor{\bsnm{Piantanida},~\bfnm{P.}\binits{P.}} \AND
  \bauthor{\bsnm{Hero},~\bfnm{A.~O.}\binits{A.~O.}}
(\byear{2015}).
\btitle{The Three-Terminal Interactive Lossy Source Coding Problem}.
\bjournal{Information Theory, IEEE Trans. on}.
\bnote{(revised)}.
\end{barticle}
\endbibitem

\bibitem{Wald-1945}
\begin{barticle}[author]
\bauthor{\bsnm{Wald},~\bfnm{Abraham}\binits{A.}}
(\byear{1945}).
\btitle{Sequential tests of statistical hypotheses}.
\bjournal{The Annals of Mathematical Statistics}
\bvolume{16}
\bpages{117--186}.
\end{barticle}
\endbibitem

\bibitem{Wald-Wolfowitz-48}
\begin{barticle}[author]
\bauthor{\bsnm{Wald},~\bfnm{Abraham}\binits{A.}} \AND
  \bauthor{\bsnm{Wolfowitz},~\bfnm{Jacob}\binits{J.}}
(\byear{1948}).
\btitle{Optimum character of the sequential probability ratio test}.
\bjournal{The Annals of Mathematical Statistics}
\bvolume{19}
\bpages{326--339}.
\end{barticle}
\endbibitem

\bibitem{Xiang-Kim-2012}
\begin{binproceedings}[author]
\bauthor{\bsnm{Xiang},~\bfnm{Yu}\binits{Y.}} \AND
  \bauthor{\bsnm{Kim},~\bfnm{Young-Han}\binits{Y.-H.}}
(\byear{2012}).
\btitle{Interactive hypothesis testing with communication constraints}.
In \bbooktitle{Communication, Control, and Computing (Allerton), 2012 50th
  Annual Allerton Conference on}
\bpages{1065-1072}.
\bdoi{10.1109/Allerton.2012.6483336}
\end{binproceedings}
\endbibitem

\bibitem{Yakowitz-1974}
\begin{barticle}[author]
\bauthor{\bsnm{Yakowitz},~\bfnm{Sidney}\binits{S.}}
(\byear{1974}).
\btitle{Multiple hypothesis testing by finite memory algorithms}.
\bjournal{The Annals of Statistics}
\bvolume{2}
\bpages{323--336}.
\end{barticle}
\endbibitem

\bibitem{Zhao-Lai-2015}
\begin{binproceedings}[author]
\bauthor{\bsnm{Zhao},~\bfnm{Wenwen}\binits{W.}} \AND
  \bauthor{\bsnm{Lai},~\bfnm{Lifeng}\binits{L.}}
(\byear{2015}).
\btitle{Distributed testing with zero-rate compression}.
In \bbooktitle{Inf. Theory, 2015 IEEE International Symposium on (ISIT)}
\bpages{2792-2796}.
\bdoi{10.1109/ISIT.2015.7282965}
\end{binproceedings}
\endbibitem

\end{thebibliography}
\end{document}